\documentclass[12pt, a4paper]{article}

\usepackage[margin=1in]{geometry} 
\usepackage{amsmath, amssymb}     
\usepackage{amsthm}               
\usepackage{algorithm}            
\usepackage{algpseudocode}        
\usepackage[numbers]{natbib}      
\usepackage{parskip}              
\usepackage{graphicx}             
\usepackage{hyperref}             
\usepackage[utf8]{inputenc}       
\usepackage[T1]{fontenc}        

\newtheorem{theorem}{\bf Theorem}[subsection]

\newtheorem{definition}{\bf Definition}[subsection]
\newtheorem{assumption}{\bf Assumption}[subsection]
\newtheorem{proposition}{\bf Proposition}[subsection]
\newtheorem{lemma}{\bf Lemma}[subsection]

\begin{document}

\title{A Gentle Introduction to Conformal Time Series Forecasting}

\author{
M. Stocker$^{1}$, W. Małgorzewicz $^{2}$, M. Fontana$^{2}$,  S. Ben Taieb $^{3,4}$
\vspace{0.5cm} \\
\small $^{1}$Karlsruhe Institute of Technology, Germany\\
\small $^{2}$Royal Holloway, University of London, United Kingdom\\
\small $^{3}$Mohamed bin Zayed University of Artificial Intelligence, Abu Dhabi, United Arab Emirates\\
\small $^{4}$University of Mons, Belgium
}

\date{} 
\maketitle

\begin{abstract}
Conformal prediction is a powerful post-hoc framework for uncertainty quantification that provides distribution-free coverage guarantees. However, these guarantees crucially rely on the assumption of exchangeability. This assumption is fundamentally violated in time series data, where temporal dependence and distributional shifts are pervasive. As a result, classical split-conformal methods may yield prediction intervals that fail to maintain nominal validity. This review unifies recent advances in conformal forecasting methods specifically designed to address nonexchangeable data. We first present a theoretical foundation, deriving finite-sample guarantees for split-conformal prediction under mild weak-dependence conditions. We then survey and classify state-of-the-art approaches that mitigate serial dependence by reweighting calibration data, dynamically updating residual distributions, or adaptively tuning target coverage levels in real time. Finally, we present a comprehensive simulation study that compares these techniques in terms of empirical coverage, interval width, and computational cost, highlighting practical trade-offs and open research directions.
\end{abstract}


\section{Overview of Conformal Prediction}
\label{sec:intro}

As underlined by the most recent reviews on the subject, the quantification and formalisation of prediction uncertainty is a key challenge in forecasting \cite{petropoulos_forecasting_2022}. A simple point forecast, $\hat{y}$, which provides a single-value estimate for a quantity $y \in \mathbb{R}$, is inherently limited. For any continuous random variable $Y$, the probability of the outcome being \emph{exactly} the point forecast is zero (i.e., $\mathbb{P}(Y=\hat{y})=0$). Therefore, the true practical value of forecasting, particularly for robust decision-making and risk assessment, lies not in identifying a single "most likely" value, but in quantifying the full spectrum of possible outcomes. The ultimate objective is to move beyond single-point estimates toward probabilistic forecasts that describe the entire predictive distribution of future quantities \cite{gneiting_probabilistic_2014, gneiting_probabilistic_2007}


The challenge, however, lies in generating these distributions accurately. Traditionally, uncertainty quantification has relied on strong parametric assumptions, such as imposing a Gaussian (Normal) distribution on the forecast errors \cite{elliott_economic_2016}. This assumption is frequently violated in real-world applications. Many phenomena, particularly in economics, finance, and anthropogenic systems, do not follow "mild" randomness. They are instead characterised by heavy tails (leptokurtosis), significant skewness, and non-linear dependencies, which the Gaussian distribution fundamentally fails to capture \cite{cont_empirical_2001, taleb_black_2007}. This mismatch often leads to a drastic underestimation of risk and a false sense of security in model predictions.

In this context, Conformal Prediction (CP) has emerged as a powerful and principled alternative to parametric methods \cite{vovk_algorithmic_2023}. It provides a distribution-free framework that wraps around any point forecasting model, from simple linear regressions to deep normalising flow models, and calibrates its predictions to produce sets with rigorous, finite-sample predictive coverage guarantees. Let $\{ (X_t, Y_t) \}_{i=1}^T$ be a sample of \(T\) random covariate/response pairs with stationary marginals. Each pair \((X_t, Y_t) = Z_t\) takes values in \(\mathcal{X} \times \mathcal{Y}\), where \(\mathcal{X}\) and \(\mathcal{Y}\) are measurable spaces. Our objective is, for given a miscoverage level $\alpha \in (0,1)$, to construct a prediction set $\mathcal{C}_{1-\alpha}(X_{T+1})$ for a new, unobserved, $Y_{T+1}$ such that the following validity property holds:
\begin{equation}
    \mathbb{P}\!\big( Y_{T+1} \in \mathcal{C}_{1-\alpha}(X_{T+1}) \big) \geq 1-\alpha 
\end{equation}
where the probability $\mathbb{P}_{\mathrm{tr}}$ is taken over  $\{ (X_t, Y_t) \}_{i=1}^T \cup Z_{T+1} $
The methodologies under the CP umbrella provide techniques and algorithms to identify regions endowed with the validity property described above. Two main families of methodologies are available: Full (or Inductive) CP and Split (or Transductive) CP. The two methodologies distinguish themselves according to how they treat the training data, and how intensive is their computational load. Apart from this general distinction, our goal in this paper is not to describe analogies and differences between the two methodologies. The interested reader can refer to the already cited main text on the subject, as well as to several introductions and reviews \cite{fontana_conformal_2023,angelopoulos_conformal_2023-5}.

We will focus our attention on Split, or Transductive CP (SCP), firstly introduced in \cite{papadopoulos_inductive_2002}, and analysed with a remarkable level of detail in \cite{lei_distribution-free_2018}. This choice is due to the specific nature of the algorithms analysed in this review, mainly focused on a Split framework.  The recipe is as follows:

First, we partition the data indices into two disjoint sets: a training set \(I_{\mathrm{train}}\) and a calibration set \(I_{\mathrm{cal}}\).
Second, we use the training set to fit our model of choice, $\hat{\mu}$, which learns a function from $\mathcal{X}$ to $\mathcal{Y}$.
Third, we define a nonconformity score function $s(x, y)$, which quantifies how "strange" or "nonconforming"\footnote{It is in principle possible to define conformity, rather than non-conformity scores, that measure "conformity" with respect to the original data} a given data pair $(x, y)$ is with respect to the training set $\{(x_t,y_t), t\in I_{\mathrm{train}}\}$. For regression problems, traditionally absolute residuals, $s(X_i, Y_i) = |Y_i - \hat{\mu}(X_i)|$ are used, but other choices like quantile-based scores (\cite{romano_conformalized_2019}) and Density/HDR-based ones (\cite{izbicki_cd-split_2022}) are possible. The key idea is that higher scores should correspond to data that fits the model poorly.
We compute the nonconformity scores for all points in the calibration set, creating a set of "typical" errors: $S_{\mathrm{cal}} = \{s_i \mid i \in I_{\mathrm{cal}}\}$.
Fifth, to achieve a target coverage of $1-\alpha$, we find the empirical $(1-\alpha)$quantile of these calibration scores. Specifically, letting $n_{\mathrm{cal}} = |I_{\mathrm{cal}}|$, we compute:
\[
\hat{q}_{1-\alpha} = \textit{Quantile}(S_{\mathrm{cal}}, \lceil (1-\alpha)(n_{\mathrm{cal}}+1) \rceil / n_{\mathrm{cal}})
\].
This $\hat{q}_{1-\alpha}$ represents the error threshold that $(1-\alpha)$ of the calibration points did not exceed.
Finally, for a new point $X_{T+1}$, we construct the prediction set by inverting the score function. We include all possible values $y \in \mathcal{Y}$ whose nonconformity score is no larger than our threshold $\hat{q}_{1-\alpha}$:
\[
\mathcal{C}_{1 - \alpha}(X_{T+1}) = \{ y \in \mathcal{Y} : s(X_{T+1}, y) \leq \hat{q}_{1-\alpha} \}
\]
When using the absolute residual score, this definition simplifies to the familiar interval: $$\mathcal{C}_{1 - \alpha}(X_{T+1}) = [\hat{\mu}(X_{T+1}) - \hat{q}_{1-\alpha}, \hat{\mu}(X_{T+1}) + \hat{q}_{1-\alpha}].$$

If the data in $I_{\mathrm{cal}}$ and the new test point $(X_{\text{test}}, Y_{\text{test}})$ are exchangeable, this simple procedure provides the powerful guarantee of finite-sample marginal coverage: $\mathbb{P}_{tr}(Y_{T+1} \in \mathcal{C}_{1 - \alpha}(X_{T+1})) \ge 1-\alpha$.

The \textit{exchangeability} assumption in this context is key This property is the key theoretical mechanism that enables the methods to provide guaranteed finite-sample coverage (i.e., validity) without making any specific parametric assumptions about the underlying data-generating process \cite{angelopoulos_theoretical_2025}.

In the absence of exchangeability, it is not trivial anymore to establish validity properties. This complication motivates the central questions driving the research in this area: Can we preserve the distribution-free guarantees of CP when predicting non-exchangeable data? How? And, in doing so, what theoretical properties or practical efficiencies must we trade off? 

This review synthesizes and unifies the modern approaches developed to solve this problem, mainly in the context of time-series data. We will collectively call this branch of CP, \emph{Conformal Forecasting}. We aim to move beyond a simple survey of algorithms. Instead, we provide a structured classification of these methods, providing an harmonised notation, and a grouping by the core philosophy they employ to address non-exchangeability.

Our contributions are threefold: (i) a practical and narrative synthesis of these baseline algorithms and their modern variants; (ii) a controlled empirical comparison that maps the validity–efficiency–compute trade-offs of these competing strategies and (iii) a deeper, unified theoretical rework of finite-sample coverage guarantees for standard SCP under checkable, weak-dependence ($\beta$-mixing) conditions; (with full proofs in the Appendixs), providing additional details with respect to the original work \cite{oliveira_split_2024}.

The review is structured as follows. Section \ref{sec:nonexTS} formalizes how time series data violate the exchangeability assumption through temporal dependence and distribution shifts. Section \ref{sec:nonexchangeable} presents theoretical guarantees for SCP under weak dependence. Section \ref{sec:adaptive_methods} reviews the four main families of adaptive conformal forecasting methods: Weighted CP (WCP), EnbPI, Adaptive CP (ACI), and Block CP (BCP), providing pseudocode for each. Section \ref{sec:emp} details and discusses the results of a comprehensive empirical study comparing these methods on simulated data. Finally, Section \ref{sec:conclusion} concludes with practical recommendations and a discussion of limitations and future research directions.

\section{Non-Exchangeability in Time Series Data}
\label{sec:nonexTS}
We argued in the previous section that the good properties of CP are fundamentally based on an exchangeability assumption. We provide a more formal definition below.

\begin{definition}[Exchangeability]\label{def:exchangeability}
A finite sequence of random elements $\{Z_i\}_{i=1}^T$ is \emph{exchangeable} if its joint law is invariant under any permutation $\pi$ in a generic set $\mathcal{P}_n$:
\[
(Z_1,\dots,Z_T)\ \stackrel{d}{=}\ (Z_{\pi(1)},\dots,Z_{\pi(T)}) \quad \text{for all permutations }\pi \in \mathcal{P}_n.
\]
\end{definition}

This assumption, although weaker than IID, is typically violated in time series settings. By definition, time series data are ordered. This temporal order is not a nuisance but the very structure that carries meaningful information. Such ordering may break exchangeability in several fundamental ways:

\begin{enumerate}
    \item \textbf{Temporal Dependence (Non-Independence):} Let $\{Z_t\}_{t=1}^T$ be a stochastic process, representing the data, defined on a probability space $(\Omega, \mathcal{F}, \mathbb{P})$. Each random variable $Z_t$ takes values in a measurable state space $(G, \mathcal{G})$ (e.g., $(\mathbb{R}, \mathcal{B}(\mathbb{R}))$, where $\mathcal{B}(\mathbb{R})$ is the Borel $\sigma$-algebra on the real line).

    The process $\{Z_t\}$ exhibits \textbf{temporal dependence} if random variables $Z_1, \dots, Z_T$ are not mutually independent. This concept is formalized by considering the information available up to a certain point in time, which is represented by the natural filtration $\mathcal{F}_{t-1} = \sigma(Z_1, \dots, Z_{t-1})$.

    A sequence of random variables is \textbf{independent} if, for all $t$, the conditional distribution of $Z_t$ given the entire past history $\mathcal{F}_{t-1}$ is identical to its marginal distribution. Formally, independence requires:
    $$
    \mathbb{P}(Z_t \in g \mid \mathcal{F}_{t-1}) = \mathbb{P}(Z_t \in g) \quad \text{for all } g \in \mathcal{G} \text{ (almost surely)}
    $$
    Conversely, the process is \textbf{temporally dependent} if this equality fails to hold for some $t$ and some set $A$. This means the past provides information for predicting $Z_t$:
    $$
    \mathbb{P}(Z_t \in g \mid \mathcal{F}_{t-1}) \neq \mathbb{P}(Z_t \in g)
    $$
    This dependence is the central feature of time series models where the past informs the future. For example:
    \begin{itemize}
        \item In \textbf{ARMA processes}, the conditional mean $\mathbb{E}[Z_t \mid \mathcal{F}_{t-1}]$ depends on $\mathcal{F}_{t-1}$.
        \item In \textbf{GARCH processes}, the conditional variance $\text{Var}(Z_t \mid \mathcal{F}_{t-1})$ depends on $\mathcal{F}_{t-1}$.
    \end{itemize}

   \item \textbf{Distribution Shift (Non-Stationarity):} We again consider the stochastic process $\{Z_t\}_{t=1}^T$ defined on $(\Omega, \mathcal{F}, \mathbb{P})$ and taking values in $(S, \mathcal{S})$. Let $P_t$ denote the marginal probability distribution (or "law") of the observation $Z_t$, defined as:
    $$
    P_t(A) = \mathbb{P}(Z_t \in A) \quad \text{for all } A \in \mathcal{S}
    $$
    A sequence of random variables is \textbf{identically distributed} if this marginal law is invariant with respect to time; that is, $P_t = P_s$ for all $t, s \in \{1, \dots, T\}$.

    The process exhibits a \textbf{distribution shift}, or is \textbf{non-stationary} (specifically, not stationary in its marginal distribution), if the "identically distributed" property fails. This means there exist at least two time points $t \neq s$ for which the marginal distributions are not equal:
    $$
    P_t \neq P_s
    $$
    This implies that for some set $A \in \mathcal{S}$, the probability $\mathbb{P}(Z_t \in A)$ is not constant in $t$. This directly violates a necessary condition for the sequence to be independent and identically distributed (i.i.d.). As exchangeability requires that the sequence be identically distributed, this condition also violates exchangeability.

    This non-stationarity can manifest in several forms:
    \begin{itemize}
        \item \textbf{Abrupt Shifts (Breaks):} The distribution changes at a specific point $T_0$, thus $P_{T_0-1}\neq P_{T_0} $
        \item \textbf{Gradual Drift:} The parameters of $P_t$ (e.g., its mean $\mathbb{E}[Z_t]$ or variance $\text{Var}(Z_t)$) evolve slowly and systematically with $t$.
        \item \textbf{Periodic Patterns (Seasonality):} The distribution follows a recurring pattern, e.g., $P_t = P_{t+k}$ for some period $k$, but $P_t \neq P_{t+1}$.
    \end{itemize}
\end{enumerate}

It is crucial to distinguish this specific temporally-structured non-exchangeability from other violations encountered in static machine learning. A common case, for instance, is covariate shift. In that setting, data $(X_t, Y_t)$ is often assumed to be independent \textit{within} the training and test sets, but the marginal distribution of the covariates $P_X$ differs between them (i.e., $P_{X_{\text{train}}} \neq P_{X_{\text{test}}}$), even if the conditional $P_{Y|X}$ remains invariant. Another violation is given by panel data settings, where data points are correlated \textit{within} a group (e.g., multiple measurements from the same patient or samples from the same batch) but the groups themselves are independent and exchangeable. In these settings, the data index is a nominal label, and the non-exchangeability arises from a latent group structure. In the time-series context, the index $t$ is fundamentally ordinal. The violations are defined by this sequence: temporal dependence relates $Z_t$ to $Z_{t-k}$, and distribution shift makes $P_t$ a function of $t$ itself. This sequential structure, governed by proximity and direction, is the unique challenge of conformal forecasting.
As a comment, since the filtration $\mathcal{F}_{t-1}$ is order-dependent, permuting the observations would give a different joint law, thus breaking exchangeability.

To better visualise these concepts, we provide some examples.

\begin{itemize}
    \item \textbf{Example 1: Lack of Independence (Temporal Dependence).}
    Consider forecasting daily temperature, a stationary autoregressive process where $Z_t$ is conditionally dependent on its history $\mathcal{F}_{t-1} = \sigma(Z_1, \dots, Z_{t-1})$. Due to the specific nature of atmospheric phenomena, temperature is a "sticky"; a very hot day ($Z_{t-1}$) is likely followed by another hot day ($Z_t$). Now, imagine our calibration set, $I_{\text{cal}}$, happens to be drawn from a long, stable period where temperatures were consistently mild. The resulting residuals in $S_{\text{cal}}$ will all be small, leading to a small threshold $\hat{q}_{1-\alpha}$. If our test set begins just after a sudden, rare heat spike, the temporal dependence means the first test point $Z_{T+1}$ is also likely to be hot. Our model, calibrated on "mild" data, will under-predict, causing a very large residual $s_{T+1}$. The sequence is not exchangeable because we cannot swap the "post-heat-spike" test point $Z_{T+1}$ with a "mild" calibration point $Z_i$ without breaking the temporal structure of the process. The test score $s_{T+1}$ is not a random draw from the same error pool as $S_{\text{cal}}$.

    \item \textbf{Example 2: Distribution Shift (Non-Stationarity).}
    Consider a model forecasting daily users for a new website, a non-stationary process where the marginal distribution $P_t$ changes. We use data from January to November (regime $P^{J,\ldots,N}$) as our calibration set $I_{\text{cal}}$ to compute our error quantile $\hat{q}_{1-\alpha}$. On December 1st, the website is featured in a viral video, and its average daily traffic permanently triples. This is a distribution shift, or a structural break, where the marginal law abruptly changes to $P^D$. Our test set, $I_{\text{test}}$, begins on December 1st. The model, calibrated on the pre-viral data , is now systematically wrong; its predictions are far too low, and the scores $s_t$ are consistently very big. The sequence is not exchangeable because the test points $\{Z_t\}_{t \in I_{\text{test}}}$ are drawn from a completely different data-generating process ($P^D$) than the calibration points $\{Z_i\}_{i \in I_{\text{cal}}}$ ($P^{J,\ldots,N}$). The temporal order is critical, and $\hat{q}_{1-\alpha}$ is not representative of anything.

    \item \textbf{Example 3: Conditional Heteroscedasticity (Volatility Clustering).}
    Consider a stochastic process characterized by time-dependent conditional variance, such as a GARCH process \cite{bollerslev_generalized_1986} , where $\sigma_t^2 = \mathrm{Var}(Y_t \mid \mathcal{F}_{t-1})$ evolves dynamically given the filtration $\mathcal{F}_{t-1}$. This is common in financial data, like daily stock returns, which exhibit \emph{volatility clustering}. "Panic" days, characterised by high variance are followed by more high-variance days, and  "calm" periods of low-variance are followed by more calm days. Suppose the calibration set $I_{\text{cal}}$ is sampled from a regime of low volatility (quiescence), resulting in a set of nonconformity scores $S_{\text{cal}}$ with low dispersion and a correspondingly small empirical quantile $\hat{q}_{1-\alpha}$. If the test set $I_{\text{test}}$ coincides with a "volatility cluster" (a period where $\sigma_t^2$ increases significantly) the magnitude of the test residuals $s_t$ will scale proportionally with the localized standard deviation. Even if the conditional mean estimator $\hat{\mu}$ remains unbiased, the fixed interval width determined by $\hat{q}_{1-\alpha}$ will be insufficient to accommodate the expanded support of the error distribution. The sequence is not exchangeable because the marginal distribution of the residuals is not invariant; the probability density of a value $s_t$ is functionally dependent on the latent volatility state at time $t$, rendering the global quantile derived from the low-variance $I_{\text{cal}}$ invalid for the high-variance test regime.
\end{itemize}

\subsection{Guarantees under Non-Exchangeability}
\label{sec:nonexchangeable}
We have described above how the defining characteristics of time series data ,namely temporal dependence and distribution shifts ,fundamentally violate the exchangeability assumption required in standard conformal prediction.

Yet, in the presence of mild violations, the validity property of Conformal is only mildly affected, and in specific cases such violation can be computed.

In the Appendix, we present a detailed theoretical analysis of this degradation, following the framework of \cite{oliveira_split_2024}. We show that under general assumptions of weak dependence (Assumptions \ref{assump:A1_app}-\ref{assump:A3_app} in the Appendix), the coverage gap, defined as the non-negative difference $G := \max\left(0, (1-\alpha) - \mathbb{P}(Y \in \mathcal{C}_{1-\alpha}(X))\right)$, representing the shortfall between the nominal coverage $1-\alpha$ and the true coverage probability, is explicitly bounded by a small constant. For any test point $i \in I_{\mathrm{test}}$, the coverage is:
\[
\mathbb{P}_{\mathrm{tr}} \left[ Y_i \in \mathcal{C}_{1 - \alpha}(X_i) \right] 
    \geq 1 - \alpha - \underbrace{(\varepsilon_{\mathrm{cal}} + \delta_{\mathrm{cal}} + \varepsilon_{\mathrm{train}})}_{\text{Slack term}}
\]
These slack terms have intuitive meanings: $(\varepsilon_{\mathrm{cal}}, \delta_{\mathrm{cal}})$ bound the concentration error (how well the $n_{\mathrm{cal}}$ calibration scores represent the "true" error distribution), while $\varepsilon_{\mathrm{train}}$ bounds the decoupling (how much dependence exists between the training set and the test point). See Theorem \ref{theo 1_app} in Appendix \ref{app:theory} for the full statement and proof.

We further show in Theorem \ref{sec:beta-mixing_app} in Appendix \ref{app:theory} how to derive explicit, non-asymptotic bounds for these slack terms in the specific case of stationary $\beta$-mixing processes (Propositions \ref{prop: 13_app}-\ref{prop: 15_app}). The $\beta$-mixing coefficient $\beta(a)$ quantifies how quickly a process "forgets" its past; if $\beta(a)$ decays rapidly, our slack terms become small, and the coverage guarantee $1-\alpha - \eta$ approaches the nominal $1-\alpha$.

The crucial takeaway is that for stationary, weakly dependent processes, standard SCP is approximately valid, and its deviations from validity are indeed very mild. The true problem is non-stationarity (distribution shift), against which these theoretical results offer limited protection. The following methods are designed to explicitly address both strong dependence and, more importantly, distribution shift.

\section{Conformal Forecasting Methods}
\label{sec:adaptive_methods}

In case we have dependencies that are more severe than the weak cases considered by \cite{oliveira_split_2024} several methodologies, with very different philosophical approaches, have appeared in the literature.

These conformal forecasting techniques are designed to restore valid coverage by explicitly handling temporal dependence and distribution shift. We first propose a narrative classification of methods according to their core philosophy, and then provide to state the algorithms in detail.
In a nutshell, methodologies can be based on:

\begin{itemize}
  \item \textbf{Reweighting:} Establish a calibration set but assign higher importance to points that are deemed more "relevant" to the current test point. This includes Weighted CP (WCP / Nex-CP, \cite{barber_conformal_2023}).
  \item \textbf{Refreshing:} Actively update the calibration set, typically using a sliding window to discard old, "stale" residuals and incorporate new ones. This includes Ensemble Batch Prediction Intervals (EnbPI, \cite{xu_conformal_2021}).
  \item \textbf{Adapting Coverage:} Establish a calibration set but dynamically update the target error rate $\alpha_t$ online, using a feedback loop to force the long-run coverage to match the user's target. This includes Adaptive Conformal Inference (ACI, \cite{gibbs_adaptive_2021,zaffran_adaptive_2022}).
  \item \textbf{Blocking:} Redefine the fundamental unit of randomization. Instead of assuming individual points are exchangeable, assume that entire blocks of data can be permuted. This includes Block CP (BCP, \cite{chernozhukov_exact_2018}).
\end{itemize}

We now proceed analysing the different groups one by 

\subsection{Weighted CP (WCP)}
\label{sec:wcp}

The first family of methods adapts to non-exchangeability by challenging the assumption that all calibration points are created equal. If the calibration set is no longer representative, perhaps we can re-weight its elements to prioritize points that are more relevant to the current prediction.

This is formalized by replacing the standard empirical quantile with a weighted empirical quantile. Given non-negative weights \(\{w_i\}_{i\in I_{\mathrm{cal}}}\) for each calibration point, we can normalized to sum to one: $\widetilde w_i = w_i / \sum_j w_j$. The weighted quantile is then defined as the smallest score $t$ that captures at least $1-\alpha$ of the weighted mass:
\[
\hat q^{(w)}_{1-\alpha}
:= \inf\Big\{ \tilde{s} \in \mathbb{R} : \sum_{i\in I_{\mathrm{cal}}} \widetilde w_i\,\mathbf{1}\{s(X_i,Y_i) \le \tilde{s}\} \ge 1-\alpha \Big\},
\]
and the resulting prediction set is $\mathcal{C}^{(w)}_{1-\alpha}(x) = \{ y : s(x,y) \le \hat q^{(w)}_{1-\alpha} \}$. Algorithm \ref{alg:NexCP_single_specific} summarises the non-exchangeable CP (Nex-CP) framework.

\begin{algorithm}[H]
\caption{Weighted Conformal Prediction}
\label{alg:NexCP_single_specific}
\begin{algorithmic}[1]
\Require Data $\{(X_t,Y_t)\}_{t=1}^T$; new covariate $X_{T+1}$; miscoverage level $\alpha$; base forecaster $\hat\mu$; Two index sets $I_{\text{train}},I_{\text{cal}}$ such that $I_{\text{train}}\cup I_{\text{cal}} = \{1,\ldots,T\}$; weights $\{w_i\}_{i\in I_{\mathrm{cal}}}$
\State Fit $\hat\mu$ on $I_{\mathrm{train}}$
\State Compute residuals $\varepsilon_i = |Y_i - \hat\mu(X_i)|$ for all $i \in I_{\mathrm{cal}}$
\State Compute normalised weights $\tilde w_i = \dfrac{w_i}{\sum_{j\in I_{\mathrm{cal}}} w_j}$ for $i\in I_{\mathrm{cal}}$
\State $\hat{q}^{(w)}_{1-\alpha} =$ empirical weighted $(1-\alpha)$-quantile of $\{\varepsilon_i\}_{i\in I_{\mathrm{cal}}}$
\State $\hat{Y}_{T+1} = \hat\mu(X_{T+1})$; \quad $\mathcal{C}^{(w)}_{1-\alpha}(X_{T+1}) = [\,\hat{Y}_{T+1} - \hat{q}^{(w)}_{1-\alpha},\ \hat{Y}_{T+1} + \hat{q}^{(w)}_{1-\alpha}\,]$
\State \Return $\mathcal{C}^{(w)}_{1-\alpha}(X_{T+1})$
\end{algorithmic}
\end{algorithm}

As shown by \cite{barber_conformal_2023}, the coverage gap of this method is bounded (see Theorem \ref{thm:wcp_app} in the Appendix). The bound is small if large weights $\widetilde{w}_i$ are assigned to calibration points $i$ that are "distributionally similar" to the test point.

The critical design choice, of course, is how to set the weights.
\begin{itemize}
    \item \textbf{Predefined Weights (Nex-CP):} The simplest approach is to use a fixed, predefined heuristic. The most common is exponential decay, $w_i \propto \rho^{t_m - t_i}$ (for $\rho \in (0,1)$), which embodies the simple idea that "the recent past matters most." A simpler version of this is the idea of thesliding window, where $w_i=1$ for the $k$ most recent points and $w_i=0$ for all others. These are computationally cheap but can fail if the process has long-term dependencies or if a past, distant event is more relevant than the recent past.
   
    \item \textbf{Conformal Risk Control (CRC):} \cite{angelopoulos_conformal_2023-1} provide a framework to generalize CP beyond the simple binary coverage loss to control any bounded, monotone risk (loss function). For example, a user might care more about the size of the interval than the coverage, or vice versa. \cite{farinhas_non-exchangeable_2024} connect this to the non-exchangeable setting by incorporating weights, providing a principled framework to choose weights that explicitly optimize the bound on the target risk, moving beyond simple heuristics.
   
    \item \textbf{Learned, Content-Based Weights:} The limitation of time-based weights is that they fail if a "heatwave" is more similar to an event from two years ago than from two days ago. We need \textit{content-based} retrieval. \textbf{Hop-CPT} \cite{auer_conformal_2023} implements this using modern Hopfield networks as an associative memory. The network stores patterns from the calibration set. When a new test point $X_t$ arrives, the network retrieves the most \textit{similar} past examples, and these similarity scores are used as the weights $w_i$. \textbf{CT-SSF} \cite{chen_conformalized_2024} achieves a similar goal by first mapping the input data $X_t$ into a "semantic feature" space $f(X_t)$ using a neural network. It then uses the network's internal attention mechanism to find calibration points with semantically similar features, using the attention scores as weights for calibration.
\end{itemize}

\subsection{Updating the Residual Distribution (EnbPI)}
\label{sec:enbpi}

Instead of re-weighting a fixed calibration set, a second class of methods adapts to new test input by actively refreshing the set of residuals $\{\hat{\varepsilon}_i\}^T_{i=1}$, in this case computed without the need of splitting training and calibration. The goal is to ensure the quantile $\hat{q}_{1-\alpha}$ is always based on the most recent, and therefore most relevant, error distribution.

The main representative is EnbPI (Ensemble Batch Prediction Intervals) by \cite{xu_conformal_2021}. EnbPI avoids the static train/calibration split by using bootstrap ensembles and out-of-bag (OOB) predictions. The procedure is as follows:

\begin{enumerate}
    \item \textbf{Training:} Instead of one model, we train an ensemble of $M$ bootstrap models (e.g., $M=25$). Each model $\hat{\mu}^{(m)}$ is fit on the data (e.g., points $1, \dots, T$) by sampling with replacement.
    \item \textbf{OOB Residuals:} For each point $i \in \{1, \dots, T\}$, some models in the ensemble did not see $(X_i, Y_i)$ during their training (they are "out-of-bag" for $i$). We create an OOB prediction $\hat{\mu}_{OOB}(X_i)$ by aggregating \textit{only} those models. This is a crucial step: it provides a quasi-out-of-sample prediction for $X_i$ from models that did not train on it, perfectly mimicking the logic of a train/calibration split.
    \item \textbf{Calibration:} We then compute the OOB residuals, $\hat{\varepsilon}_i$ for all points: $\hat{\varepsilon}_i = |Y_i - \hat{\mu}_{OOB}(X_i)|$ for all $i=1, \dots, T$. This collection $\{\hat{\varepsilon}_i\}_{i=1}^T$ forms our initial calibration set.
    \item \textbf{Prediction:} For a new test point $X_{T+1}$, we get a prediction $\hat{\mu}(X_{T+1})$ by aggregating \textit{all} $B$ models (since none have seen this point). We then compute the quantile $\hat{q}_{1-\alpha}$ from our pool of OOB residuals $\{\hat{\varepsilon}_i\}_{i=1}^T$. The interval is $[\hat{\mu}(X_{T+1}) - \hat{q}, \hat{\mu}(X_{T+1}) + \hat{q}]$.
    \item \textbf{Updating:} This is the key adaptation. EnbPI is run in a sliding window. After a "batch" of $\delta$ new points are observed, their new OOB residuals are computed and added to the pool, while the $s$ \textit{oldest} residuals are discarded. This ensures the residual pool "refreshes" and gradually forgets the distant past, allowing it to adapt to distribution shifts.
\end{enumerate}

EnbPI provides approximate marginal coverage under stationarity and mixing conditions \cite{xu_conformal_2021}. Its main trade-off is computational: it requires training and storing $B$ models and, in its sequential form, re-calculating OOB residuals, which is significantly more expensive than SCP or WCP.

\begin{algorithm}[H]
\caption{EnbPI (Ensemble Batch Prediction Intervals)}
\label{alg:EnbPI_M}
\begin{algorithmic}[1]
\Require Data $\{(X_i,Y_i)\}_{i=1}^{T}$; base forecaster $\hat\mu$; miscoverage $\alpha$; aggregation method $\varphi$; number of resamples $M$; batch size $\delta$; test data $\{(X_t,Y_t)\}_{t=T+1}^{T+T_1}$ with $Y_t$ revealed only after each batch of size $\delta$ is constructed.
\For{$m=1 \rightarrow M$}
  \State Sample with replacement an index multiset $I_m$ from $\{1,\dots,T\}$.
  \State Fit $\hat\mu^{(m)}$ on $\{(X_i,Y_i): i\in I_m\}$.
\EndFor
\State \textbf{Initialize $\varepsilon = \{\}$:}
\For{$i=1\rightarrow T$}
  \State $O_i \leftarrow \{\,m:\ i\notin S_m\,\}$ \hfill (OOB models for $i$)
  \State \textbf{if} $O_i\neq\emptyset$ \textbf{then} $\hat \mu(X_i) = \varphi\big(\{\hat\mu^{(m)}(X_i): m\in O_i\}\big)$ 
  \par \textbf{else} $\hat \mu(X_i) = \varphi\big(\{\hat\mu^{(m)}(X_i)\}_{m=1}^M\big)$
  \State $\varepsilon_i = |Y_i- \hat\mu(X_i)|$; \State \textbf{Update} $\varepsilon = \varepsilon \ \cup \ \{\varepsilon_i$\}
\EndFor
\State \textbf{Initialize} $\mathcal{C} = \{\}$
\For{$t=T+1 \rightarrow T+T_1$}
  \State $\hat \mu(X_t) = \varphi \left( \{\hat\mu^{(m)}(X_t)\}_{m=1}^M \right)$
  \State $\hat{q}_{1-\alpha, t} = (1-\alpha)$ quantile of $\varepsilon$
  \State $\mathcal{C}_{1-\alpha}(X_t) = \big[\hat \mu(X_t)-\hat{q}_{1-\alpha, t},\ \hat \mu(X_t)+\hat{q}_{1-\alpha, t}\big]$
  \State \textbf{Update} $\mathcal{C} =  \mathcal{C} \ \cup \ \{\mathcal{C}_{1-\alpha}(X_t)\}$
  \If{$(t-T)\bmod s = 0$}
    \For{$j=t-s \rightarrow t-1$}
      \State observe $Y_j$;\quad compute    $\ \varepsilon_j = |Y_j-\hat \mu(X_j)|$
      \State \textbf{Update} $\varepsilon = (\varepsilon - \{\varepsilon_1\}) \cup \{\varepsilon_j\}$ and reset index of $\varepsilon$
    \EndFor
  \EndIf
\EndFor
\State \textbf{return} $\mathcal{C}$
\end{algorithmic}
\end{algorithm}

This approach has been powerful and influential, inspiring several key extensions. A major limitation of EnbPI is that it ignores heteroskedasticity: it assumes all residuals are drawn from the same (shifting) pool. It produces a single interval width for all $X_t$.
SPCI (Sequential Predictive Conformal Inference) by \cite{xu_sequential_2023} fixes this. It replaces the global residual quantile with a conditional one. Instead of asking "How big are residuals on average?", it asks, "Given the features $X_t$, how big is the residual likely to be?" It uses Quantile Random Forests (QRF) to learn a function $\hat{Q}_{\tau}(\varepsilon \mid X_t)$ that maps features to a residual quantile. This allows the interval to be naturally wider for volatile inputs and tighter for stable inputs, providing a much sharper and more adaptive forecast.
This idea was further modernized by SPCI-T \cite{lee_transformer_2024}, which replaces the QRF module with a state-of-the-art Transformer architecture, allowing the conditional quantile model to capture more complex and long-range temporal dependencies.

\subsection{Updating the Coverage Rate via Adaptation}
\label{sec:aci}

The final family of methods takes a different approach. It keeps the indices of the original calibration set $I_{\text{cal}}$ fixed, unlike \emph{EnbPI}, and does not introduce modifications to the quantile computation, unlike \emph{WCP}. What it updates online is the target miscoverage level $\alpha$ itself.

The first example of these algorithms is ACI (Adaptive Conformal Inference), proposed by \cite{gibbs_adaptive_2021} for sequential settings where the true $Y_t$ is revealed after each prediction. It maintains a running "effective" miscoverage level, $\alpha_t$, which it updates at every time step based on its past performance. It is, in essence, a feedback controller.

At time $t$, the algorithm produces an interval $\mathcal{C}_{1-\alpha_t}(X_t)$ using the current level $\alpha_t$ and the fixed calibration set $\{Z_i,i\in I_{\text{cal}}\}$. After observing the true $Y_t$, it updates the level for the next step using a simple additive rule:
\begin{equation}
\label{ACI_eq}
\alpha_{t+1} \ =\ \alpha_t + \gamma \Bigl(\alpha - \mathbf{1}\{ Y_t \notin {\mathcal C}_{1-\alpha_t}(X_t)\}\Bigr),
\end{equation}
where $\gamma > 0$ is a step-size parameter, or learning rate.

The logic is intuitive. $\alpha$ is the target level, and for each time step, the algorithm observes a binary error $\mathbf{1}\{Y_t \notin \mathcal{C}\}$.
\begin{itemize}
    \item \textbf{If miscoverage occurs} ($Y_t \notin \mathcal{C}$): The term in parentheses is $(\alpha - 1)$, which is negative. $\alpha_{t+1}$ decreases. This means $1-\alpha_{t+1}$ \textit{increases} (e.g., from 90\% to 91\%). The next quantile $\hat{q}_{1-\alpha}$ will be larger, and the next interval wider. The system self-corrects by becoming more conservative.
    \item \textbf{If coverage occurs} ($Y_t \in \mathcal{C}$): The term is $(\alpha - 0) = \alpha$, which is positive. $\alpha_{t+1}$ \textit{increases}. This means $1-\alpha_{t+1}$ decreases (e.g., from 90\% to 89.9\%). The next quantile will be smaller, and the interval tighter. The system self-corrects by becoming more efficient.
\end{itemize}
This feedback loop (see Algorithm \ref{alg:ACI}) provides a guarantee that is different from the classical CP guarantee: the long-run empirical miscoverage rate is proven to converge to $\alpha$, regardless of the data-generating process (see Theorem \ref{thm:aci_app} in Appendix B). It can handle arbitrary distribution shifts, as long as feedback is provided.

\begin{algorithm}[H]
\caption{ACI (Adaptive Conformal Inference)}
\label{alg:ACI}
\begin{algorithmic}[1]
\Require Data $\{(X_t,Y_t)\}_{t=1}^T$; sequential test data $\{(X_t, Y_t)\}_{t=T+1}^{T+T_1}$; miscoverage $\alpha$; base forecaster $\hat\mu$; step size $\gamma$; Index sets $I_{\text{train}}, I_{\text{cal}}$
\State Fit $\hat\mu$ on $I_{\mathrm{train}}$
\State Compute fixed calibration scores $\mathcal{E}_{\mathrm{cal}} = \{|Y_i - \hat\mu(X_i)| \mid i \in I_{\mathrm{cal}}\}$
\State Let $n_{\mathrm{cal}} = |I_{\mathrm{cal}}|$
\State Initialize effective level $\alpha_{T+1} = \alpha$
\State Initialize prediction sets $\mathcal{C} = \{\}$
\For{$t = T+1$ \textbf{to} $T+T_1$}
    \State \Comment{Compute quantile using the *current* level $\alpha_t$}
    \State $\hat{q}_{1-\alpha_t} = \textit{Quantile}(\mathcal{E}_{\mathrm{cal}}, \lceil (1-\alpha_t)(n_{\mathrm{cal}}+1) \rceil / n_{\mathrm{cal}})$
    \State Get prediction $\hat{Y}_t = \hat\mu(X_t)$
    \State Construct interval $\mathcal{C}_{1-\alpha_t}(X_t) = [\hat{Y}_t - \hat{q}_{1-\alpha_t}, \hat{Y}_t + \hat{q}_{1-\alpha_t}]$
    \State Append $\mathcal{C}_{1-\alpha_t}(X_t)$ to $\mathcal{C}$
    \State \Comment{Wait for true $Y_t$ to be revealed to get feedback}
    \State Compute error $e_t = \mathbf{1}\{ Y_t \notin \mathcal{C}_{1-\alpha_t}(X_t)\}$
    \State Update level for next step: $\alpha_{t+1} = \alpha_t + \gamma (\alpha - e_t)$
\EndFor
\State \Return $\mathcal{C}$
\end{algorithmic}
\end{algorithm}

The main practical challenge of ACI is choosing the step size $\gamma$. A large $\gamma$ (high learning rate) adapts very quickly to shifts but can be unstable and oscillate wildly. A small $\gamma$ is stable and converges smoothly but adapts too slowly to abrupt shifts. Several extensions have been proposed to solve this:
\begin{itemize}
    \item \textbf{AgACI:} \cite{zaffran_adaptive_2022} (Aggregated ACI) removes the need to pick one $\gamma$ by running multiple ACI "experts" in parallel, each with a different $\gamma_k$. It then uses an online expert aggregation algorithm (like Bernstein Online Aggregation) to form a weighted average of their outputs, adaptively trusting the "expert" $\gamma_k$ that has performed best on the recent past.
   
    \item \textbf{Time-Dependent Step Sizes:} \cite{angelopoulos_online_2024} propose a more standard optimization approach, using a time-dependent $\gamma_t$ that decays (e.g., $\gamma_t \propto 1/\sqrt{t}$). This allows the algorithm to be responsive at the beginning and stabilize as it gathers more data, though it may be slow to react to a late-stage shift.
   
    \item \textbf{Conformal PID Control:} \cite{angelopoulos_conformal_2023-2} provides the most sophisticated update rule by reframing ACI as a control problem. The update includes Proportional (P), Integral (I), and Derivative (D) terms. The P-term is standard ACI (reacts to current error). The I-term sums past errors (reacts to systematic bias, e.g., "we've been under-covering by 2\% for 50 steps"). The D-term is a "scorecaster" that anticipates error (e.g., "it's Friday, and errors are always higher on Fridays"). This feed-forward component is especially powerful for handling predictable shifts like seasonality.
\end{itemize}

\subsection{Block CP (BCP)}
\label{sec:bcp}

A final, conceptually distinct approach is Block CP (BCP), introduced in \cite{chernozhukov_exact_2018}. This method redefines the object of randomization, positing that while individual data points $Z_t$ are not exchangeable, entire blocks of data may be (approximately) exchangeable.

The original transductive formulation is computationally intensive, as it requires re-estimating the model for every candidate future and every permutation. This procedure is generally considered "unfeasible" in high-dimensional or functional data settings \cite{ajroldi_conformal_2023}. A scalable and computationally efficient alternative is the Split-Conformal BCP, which adapts the blocking scheme to the inductive (split) framework \cite{ajroldi_conformal_2023, diquigiovanni_distribution-free_2024}.

This split-BCP procedure modifies the standard BCP algorithm as follows:

\begin{enumerate}
    \item \textbf{Split and Train:} The data is partitioned into a training set, with indices $I_{\mathrm{train}}$ and a calibration set, with indices$I_{\mathrm{cal}}$. A model $\hat{\mu}$ is fit once on the training set and is then held fixed.
   
    \item \textbf{Define Blocks and Permutations:} A block size $B$ is selected. A family of permutations, $\Pi$, is defined to act only on the indices of the calibration set $I_{\mathrm{cal}} \cup \{T+1\}$. The training set $I_{\mathrm{train}}$ remains invariant under all permutations \cite{ajroldi_conformal_2023, diquigiovanni_distribution-free_2024}.
   
    \item \textbf{Define Nonconformity Score:} A nonconformity score $s(x,y)$ is defined based on the fixed model $\hat{\mu}$ (e.g., $s(X_t, Y_t) = |Y_t - \hat{\mu}(X_t)|$).
   
    \item \textbf{Calculate p-value:} For a candidate future $y$, its nonconformity score is $R_{T+1} = s(X_{T+1}, y)$. A set of permuted scores $\{R_{\pi}\}$ is generated by applying the permutations $\pi \in \Pi$. The p-value, $\hat{p}(y)$, is the fraction of permuted scores that are greater than or equal to the score of the candidate:
    \[
    \hat p(y) \;=\; \frac{1}{|\Pi|} \sum_{\pi \in \Pi}
    \mathbf{1}\!\left\{ R_{\pi} \geq R_{T+1} \right\}.
    \]
    (Here, $R_{\pi}$ refers to the score of the permuted block that lands at the test position, $Z_{\pi(T+1)}$, which is evaluated using the fixed model $\hat{\mu}$ trained on $I_{\mathrm{train}}$ \cite{ajroldi_conformal_2023}).
   
    \item \textbf{Prediction Set:} The final $(1-\alpha)$ CP set is formed by inverting this p-value test:
    \[
    \mathcal{C}_{1-\alpha}^B = \{\, y : \hat p(y) > \alpha \,\}.
    \]
\end{enumerate}

This split-BCP approach loses the exact finite-sample validity of the transductive method. However, it retains robust theoretical guarantees, providing approximate validity and asymptotic exactness under weak dependence conditions (e.g., strong mixing or ergodicity) \cite{ajroldi_conformal_2023, diquigiovanni_distribution-free_2024}. A significant practical advantage is that, with an appropriate choice of nonconformity score, the set $\mathcal{C}_{1-\alpha}^B$ can often be computed in closed form, avoiding the infeasible search over all possible $y$ \cite{ajroldi_conformal_2023, diquigiovanni_distribution-free_2024}.

\section{A Simulation-Based Experimental Comparison}
\label{sec:emp}

We now empirically compare the main conformal methods (SCP, WCP, ACI, and EnbPI, SCP-Block) on simulated time series data. We evaluate each on test coverage, average interval width, and computational cost. The code can be found in \cite{malgorzewicz_intro--cp-for-time-series-forecasting_2025}.

\subsection{Data-Generating Processes}
We study four canonical processes: two stationary, $\beta$-mixing benchmarks, one non-exchangeable process designed to break the methods, and one heteroscedastic model. We generate $n=900$ pairs $(X_t,Y_t)$, split into train/cal/test of (300, 300, 300), and repeat $R=50$ times. The covariate $X_t = (Y_{t-1},\ldots,Y_{t-p})$ collects past lags.

\begin{itemize}
    \item \textbf{AR(1)}: A simple, weakly dependent process.
    \[
      Y_t = 0.8\,Y_{t-1} + \varepsilon_t, \qquad \varepsilon_t \sim \mathcal{N}(0,1).
    \]
    \item \textbf{ARMA(1,1)}: A stationary process with slightly more complex memory.
    \[
      Y_t = 0.5\,Y_{t-1} + \varepsilon_t + 0.4\,\varepsilon_{t-1}, \qquad \varepsilon_t \sim \mathcal{N}(0,1).
    \]
    \item \textbf{Mean shift} (non-exchangeable): A process with an abrupt, permanent distribution shift. The shift occurs at $t^\star = 601$, the first point in the test set.
    \[
      Y_t = \mu_t + \varepsilon_t,\quad 
      \mu_t = \begin{cases}
        \mu_0, & t \le 600,\\
        \mu_0+1.0, & t > 600,
      \end{cases}
      \qquad \varepsilon_t \sim \mathcal{N}(0,1).
    \]
    \item \textbf{GARCH(1,1)}: A heteroscedastic process with an autoregressive structure within itself.
    \[
      Y_t = \varepsilon_t \sqrt{(0.3+0.5 Y^2_{t-1} + 0.1)}
      \qquad \varepsilon_t \sim \mathcal{N}(0,1).
    \]
\end{itemize}

\paragraph{Base forecaster}
To isolate the effect of the conformal post-processing, all methods use the same simple autoregression fit by least squares (AR-LS). The model $\hat Y_t = \langle a, X_t\rangle$ is fit \textit{once} on the training block $I_{\mathrm{train}}$ and its coefficients $a$ are held fixed. This is a crucial design choice: the forecaster itself does \textit{not} adapt to the mean shift, forcing the conformal layer to do all the work.

\subsection{Methods and Metrics}
All methods use absolute residual scores $|Y_t-\hat\mu(X_t)|$ with target coverage $1-\alpha=0.9$.

\begin{itemize}
    \item \textbf{SCP}: Standard split conformal, quantile computed once on $I_{\mathrm{cal}}$ (points 301-600).
    \item \textbf{Blocked SCP}: Split conformal prediction using non-overlapping blocks, with block sizes $B\in\{2,3\}$
    \item \textbf{WCP} (Nex-CP): Three fixed weighting schemes on $I_{\mathrm{cal}}$: (1) Exponential decay ($\rho=0.99$), (2) Linear ramp (more weight to recent), (3) Sliding window (only last 50 points, 551-600).
    \item \textbf{EnbPI}: $B=25$ bootstrap AR-LS models (trained on 1-300, OOB residuals on 301-600); mean aggregation; sliding residual pool with refresh frequency $s\in\{1,10,100\}$.
    \item \textbf{ACI}: Step sizes $\gamma\in\{0.001,\,0.005,\,0.01\}$; $\alpha$ updated at each test step, quantile recomputed from the fixed $I_{\mathrm{cal}}$ (301-600).
\end{itemize}

\paragraph{Metrics}
For each method on the test block (points 601-900), we compute: (1) \textbf{Coverage} (empirical frequency $\frac{1}{300}\sum \mathbf{1}\{Y_t \in \mathcal{C}_t\}$), (2) \textbf{Average width} ($\frac{1}{300}\sum |\mathcal{C}_t|$), and (3) \textbf{Wall time}.

\subsection{Results and Discussion}

We plot mean coverage vs. mean width for each process (Figure \ref{fig:vis}) and a single bar chart for average runtime (Figure \ref{fig:vis-runtime}). For each process, we additionally plot the associated error bars, using a 95\% confidence interval around the mean coverage and mean width. 

The empirical results distinguish the performance of the conformal strategies across the different data-generating processes.

In the stationary settings (AR(1), ARMA(1,1), and GARCH(1,1)), the results are consistent. As shown in Figure \ref{fig:vis}, most tested methods, including SCP, all WCP variants, all ACI variants, and all EnbPI variants, achieve empirical coverage very close to the nominal $0.9$ target. The notable exception is SCP-block, which visibly under-covers in all three stationary scenarios, failing to reach the nominal target. Among the valid methods, primary differences are in statistical efficiency (interval width). EnbPI consistently produces the widest intervals, a likely result of the variance introduced by its bootstrap-based procedure. In contrast, SCP, WCP, and ACI are more efficient, yielding tighter intervals of comparable widths. For stable, stationary processes, these results suggest the baseline SCP is a sufficient, efficient, and valid method.

The non-stationary setting (Mean-Shift) reveals significant performance disparities. Here, several methods fail. SCP's coverage degrades to approximately 0.84, and SCP-block's coverage falls to \~0.81-0.84. The WCP-window method also fails, with coverage dropping to ~0.81. Their calibrated quantiles, $\hat{q}_{1-\alpha}$, are computed from the pre-shift calibration data ($I_{\mathrm{cal}}$) and are therefore invalid estimates for the post-shift regime. These methods, blind to the model's new systematic error, continue to produce overly narrow and invalid intervals.

In contrast, the other adaptive methods successfully handle the abrupt shift. ACI (all $\gamma$ values), EnbPI (all $s$ values), WCP-exp, and WCP-linear all maintain coverage at or near the nominal 0.9 level. Their success is attributable to their explicit adaptation mechanisms. ACI uses an active feedback loop to widen intervals. EnbPI's sliding window refreshes its residual pool, while WCP-exp and WCP-linear successfully

\begin{figure}[!htbp]
  \centering
  \includegraphics[width=1\linewidth]{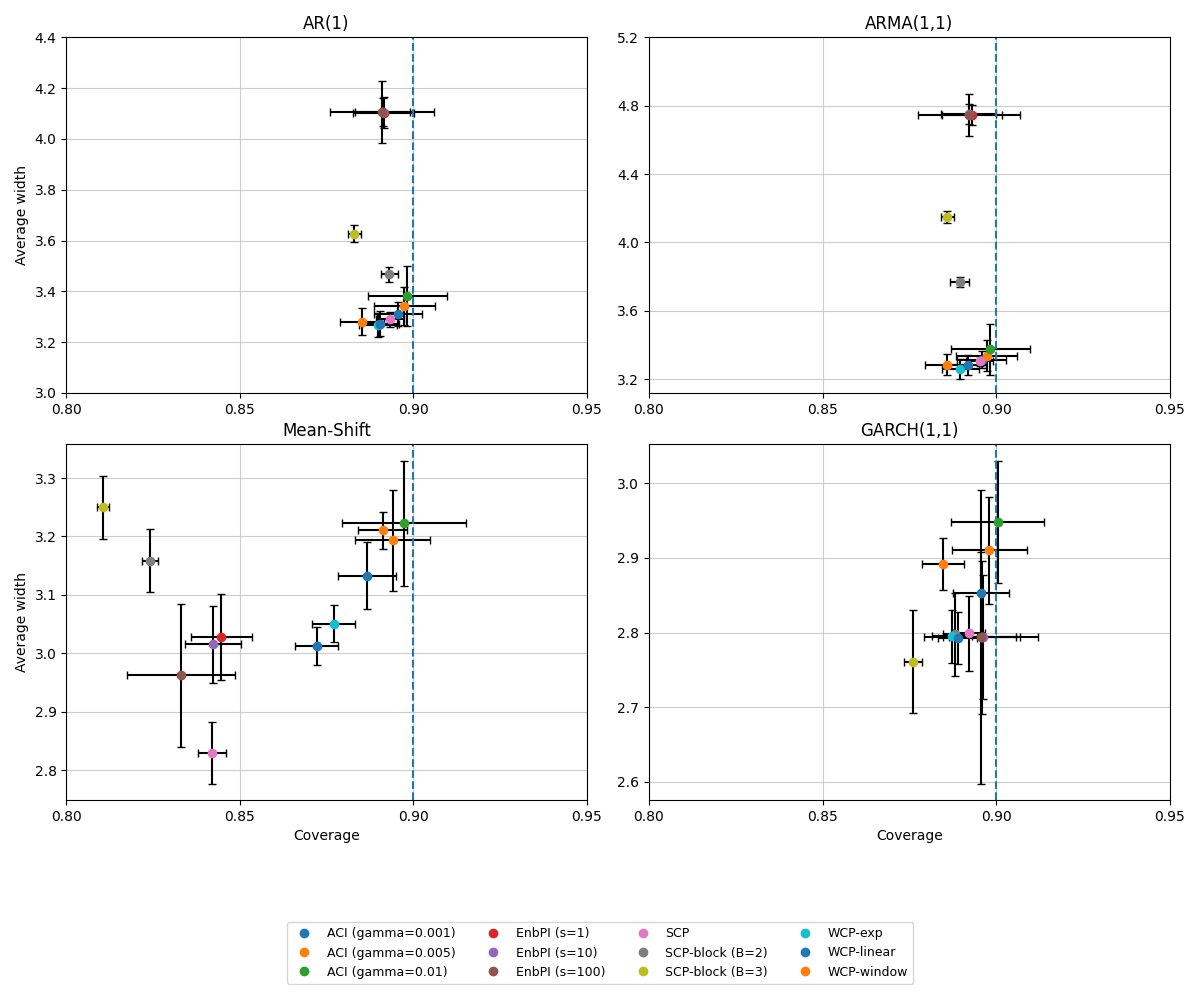}
  \caption{Coverage vs. width for different Data Generating Processes. The vertical line marks the $1-\alpha=0.9$ target, and the error bars a confidence interval of 95\%.}
  \label{fig:vis}
\end{figure}

\begin{figure}[htbp]
  \centering
  \includegraphics[width=.8\linewidth]{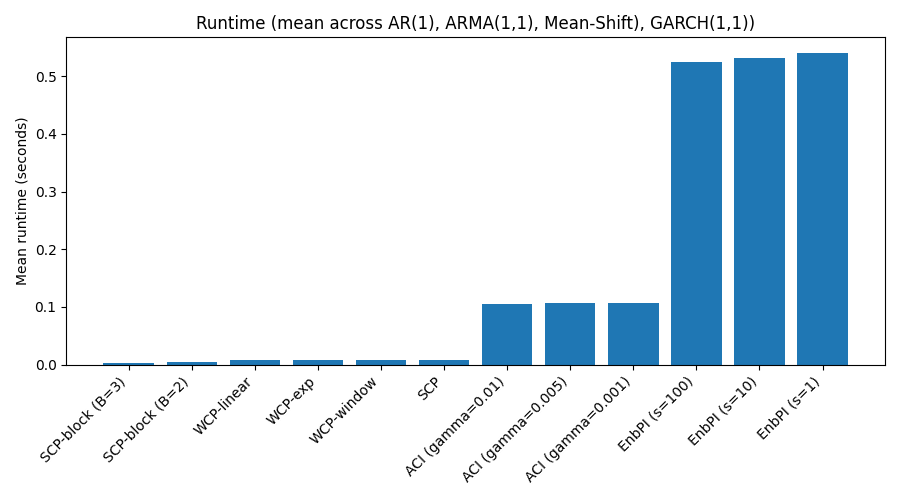}
  \caption{Average runtime by method (aggregated across processes and runs).}
  \label{fig:vis-runtime}
\end{figure}

\section{Conclusion}
\label{sec:conclusion}

Classical CP relies on the assumption of exchangeability, which is fundamentally violated by time series data due to temporal dependence and distribution shifts. This review synthesizes and evaluates modern conformal forecasting methods designed to address this limitation.

The findings indicate that methods adapt via four primary mechanisms: reweighting calibration data (WCP), refreshing the residual pool (EnbPI), adapting the target coverage level online (ACI), or blocking the data (SCP-block). The theoretical analysis (Appendix \ref{app:theory}) confirms that for weakly-dependent, stationary ($\beta$-mixing) processes, standard SCP provides approximately valid coverage. The empirical study supports this, showing that SCP, WCP, ACI, and EnbPI all achieve nominal coverage on stationary data. However, the study also reveals two key failures: (1) SCP-block failed to provide valid coverage even in the simple stationary settings, and (2) standard SCP and some adaptive variants (like WCP-window) fail under an abrupt distribution shift. In such non-stationary settings, methods with explicit recency-focused adaptation (ACI, EnbPI, and WCP-exp/linear) were all shown to successfully restore nominal coverage.

\subsection*{Recommendations for Practice}
Based on this analysis, the choice of method should be guided by the properties of the data stream and operational constraints.
\begin{itemize}
    \item For stable, stationary processes with weak dependence, SCP provides a valid, efficient, and computationally inexpensive baseline. The added complexity of most adaptive methods is unnecessary.
   
    \item When non-stationarity (e.g., abrupt shifts) is anticipated, several methods are robust. The choice involves a clear trade-off between speed and complexity:
    \begin{itemize}
        \item \textbf{WCP} (with exponential or linear decay) offers a highly practical solution. It was shown to be robust to the shift while remaining as computationally inexpensive as standard SCP.
        \item \textbf{ACI} is also a robust method, using active feedback to maintain coverage. It is moderately more costly, as it requires re-calculating quantiles at each step.
        \item \textbf{EnbPI} is robust but comes with a significant computational cost (training an ensemble) and tends to produce wider intervals. Its conditional variants (like SPCI) are powerful but should only be considered if this computational overhead is acceptable.
        item 
        \item \textbf{SCP-block} showed weak performance in this specific study, failing to achieve nominal coverage even in the stationary, $\beta$-mixing scenarios. Its practical application may require more careful tuning of block size.
    \end{itemize}
   
\end{itemize}

\subsection*{Limitations and Future Directions}
The review of the theory, while not being exhaustive, serves as a good starting point in order to formalise and systematise the growing body of knowledge on Conformal Forecasting. The simulation study shows few applicative cases, and serves as minimalistic baseline comparison to start comparing different methods. Such minimalism highlight avenues for future research. 

First, the use of a simple, fixed AR-LS model was intended to isolate the effect of the conformal layer. In practice, more complex prediction methods (e.g., re-trained ARIMA or neural models) would reduce the magnitude of the residuals, improving the efficiency of all methods and potentially altering their relative performance. 
In addition, more sophisticated scoring functions could be employed—ones capable of producing not only sharper prediction regions but also asymptotic conditional validity properties (see, e.g., \cite{izbicki_cd-split_2022, chernozhukov_distributional_2021})—and evaluated empirically.


Secondly, the empirical evaluation was limited to simple processes. Further testing is required on data with more complex structures, strong seasonality, long memory and multi-step forecast horizons. 
Thirdly, we mainly focused our attention to the univariate case. The issue of multivariate time series forecasting in a Conformal Setting joins two very new areas of research (see e.g. \cite{dheur_unified_2025}). all simulation examples are assumed to be stationary. 
The deviations from exchangeability considered in this analysis are also fairly simple, and an exploration of methods for conformal forecasting data with locally stationary, or plainly nonstationary time series data. Third, the hyperparameter sweeps were coarse; the performance of ACI, EnbPI, WCP, and SCP-block is sensitive to their respective parameters ($\gamma$, $s$, $\rho$, $B$), and a comprehensive optimization study would be required to establish a definitive performance ranking. Finally, this study focused on canonical baselines and did not empirically test more advanced variants, such as learned-weight WCP (e.g., Hop-CPT) or advanced ACI (e.g., AgACI, PID control), which remain promising areas for benchmarking

The outlook for conformal forecasting involves the development of hybrid methods that combine these adaptive strategies. Examples include integrating ACI feedback loop with SPCI conditional quantile estimation, or pairing WCP with weights learned with more advanced models, as well as the identification of fixed or learned weight schemes for more complex dependency stuctures. As data complexity increases, the demand for computationally tractable uncertainty quantification that is robust to non-exchangeability will continue to grow.

Moreover, while our attention has been focused on the classical, interval prediction setting, very interesting and recent results (\cite{allen_-sample_2025}) pave the way for extensions of Conformal forecasting methods to distributional forecasting.

\section*{Acknowledgments}
The Authors would like to thank Prof. Dr. Mathias Trabs and Prof. Rainer Von Sachs for comments on early drafts of this manuscript.

\clearpage

\appendix


\section{Theoretical Guarantees for SCP under Weak Dependence}
\label{app:theory}

This appendix provides the detailed theoretical results summarized in Section \ref{sec:nonexchangeable}. Unless otherwise stated, the assumptions, theorems, definitions, propositions and lemmas follow \cite{oliveira_split_2024}, of which we provide a deeper and clearer explanation; we adapt notation where necessary.

\subsection{Setup and Notation}

We consider the following setting:
\begin{itemize}
    \item Let \((X_i, Y_i)_{i=1}^n\) be a sample of \(n\) random covariate/response pairs with stationary marginals.
    \item We also consider an independent random pair \((X_*, Y_*)\) (independent of the sample) such that \((X_i, Y_i) \stackrel{d}{=} (X_*, Y_*)\) for all \(i \in [n]\).
    \item Let \(s : (\mathcal{X} \times \mathcal{Y})^{n_{\mathrm{train}} + 1} \to \mathbb{R}\) be a function defining a nonconformity score. For any \((x, y)\), define:
    \[
    \hat{s}_{\mathrm{train}}(x, y) := s\left( (X_i, Y_i)_{i \in I_{\mathrm{train}}}, (x, y) \right).
    \]
    \item For \(\alpha \in (0,1)\), define the empirical quantile of the calibration scores:
    \begin{equation}
    \hat{q}_{1-\alpha,\mathrm{cal}} := \inf \left\{ t \in \mathbb{R} : \frac{1}{n_{\mathrm{cal}}} \sum_{i \in I_{\mathrm{cal}}} \mathbf{1} \left\{ \hat{s}_{\mathrm{train}}(X_i, Y_i) \leq t \right\} \geq 1-\alpha \right\}.
    \label{eq:q_cal_app}
    \end{equation}
    \item For any \(x \in \mathcal{X}\), define the prediction set:
    \[
    \mathcal{C}_{1-\alpha}(x) := \left\{ y \in \mathcal{Y} : \hat{s}_{\mathrm{train}}(x, y) \leq \hat{q}_{1-\alpha,\mathrm{cal}} \right\}.
    \]
    \item Let $\mathcal{F}_{\mathrm{train}} := \sigma\big((X_i,Y_i)_{i\in I_{\mathrm{train}}}\big)$. We use $\mathbb{P}_{\mathrm{tr}}(\cdot) := \mathbb{P}(\cdot \mid \mathcal{F}_{\mathrm{train}})$ to denote probabilities conditional on the training data.
    \item Define the true conditional CDF:
    \begin{equation}
    P_{q_{\mathrm{train}}} := \mathbb{P}\left[ \hat{s}_{\mathrm{train}}(X_*, Y_*) \leq q_{\mathrm{train}} \mid \mathcal{F}_{\mathrm{train}} \right].
    \label{eq:p_qtrain_app}
    \end{equation}
\end{itemize}

\subsection{Main Assumptions for Weak Dependence}

We assume the following conditions hold, which enable coverage guarantees under mild non-exchangeability.

\begin{assumption}[A1: Calibration concentration]\label{assump:A1_app}
There exist \(\delta_{\mathrm{cal}} \in (0,1)\) and \(\varepsilon_{\mathrm{cal}} \in (0,1)\) such that, for every training-dependent threshold $q_{\mathrm{train}}$:
\begin{equation}
    \mathbb{P}_{\mathrm{tr}} \left( \left| \frac{1}{n_{\mathrm{cal}}} \sum_{i \in I_{\mathrm{cal}}} \mathbf{1} \left\{ \hat{s}_{\mathrm{train}}(X_i, Y_i) \leq q_{\mathrm{train}} \right\} - P_{q_{\mathrm{train}}} \right| \leq \varepsilon_{\mathrm{cal}} \right) \geq 1 - \delta_{\mathrm{cal}}.
    \label{eq:cal-concentration_app}
\end{equation}
\end{assumption}

This states that the empirical CDF of calibration scores concentrates around the true CDF.

\begin{assumption}[A2: Test decoupling]\label{assump:A2_app}
There exists \(\varepsilon_{\mathrm{train}} \in (0,1)\) such that, for all \(i \in I_{\mathrm{test}}\) and every training-dependent threshold $q_{\mathrm{train}}$:
\begin{equation}
    \left| \mathbb{P}_{\mathrm{tr}} \left[ \hat{s}_{\mathrm{train}}(X_i, Y_i) \leq q_{\mathrm{train}} \right] - \mathbb{P}_{\mathrm{tr}} \left[ \hat{s}_{\mathrm{train}}(X_*, Y_*) \leq q_{\mathrm{train}} \right] \right| \leq \varepsilon_{\mathrm{train}}.
    \label{eq:test-decoupling_app}
\end{equation}
\end{assumption}
This ensures the test scores behave similarly to an independent score, bounding the dependence between train and test sets.

\begin{assumption}[A3: Test concentration]\label{assump:A3_app}
There exist \(\delta_{\mathrm{test}} \in (0,1)\) and \(\varepsilon_{\mathrm{test}} \in (0,1)\) such that, for every training-dependent threshold $q_{\mathrm{train}}$:
\begin{equation}
    \mathbb{P}_{\mathrm{tr}} \left( \left| \frac{1}{n_{\mathrm{test}}} \sum_{i \in I_{\mathrm{test}}} \mathbf{1} \left\{ \hat{s}_{\mathrm{train}}(X_i, Y_i) \leq q_{\mathrm{train}} \right\} - P_{q_{\mathrm{train}}} \right| \leq \varepsilon_{\mathrm{test}} \right) \geq 1 - \delta_{\mathrm{test}}.
    \label{eq:test-concentration_app}
\end{equation}
\end{assumption}
This is the same as (A1) but for the test set.

\paragraph{Conditional Assumptions.}
For conditional coverage, we require analogous assumptions, (A4) and (A5), that hold uniformly over subsets $A \in \mathcal{A}$.

\begin{assumption}[A4: Conditional calibration concentration]\label{assump:A4_app}
There exist $\varepsilon_{\mathrm{cal}}, \delta_{\mathrm{cal}} \in (0,1)$ such that, for every $q_{\mathrm{train}}$:
\[
\mathbb{P}_{\mathrm{tr}}\!\left[
    \sup_{A \in \mathcal{A}}
    \left|
        \frac{1}{\max\{n_{\mathrm{cal}}(A),1\}}
        \sum_{i \in I_{\mathrm{cal}}(A)}
            \mathbf{1}\{\hat{s}_{\mathrm{train}}(X_i,Y_i) \leq q_{\mathrm{train}}\}
        - P_{q,\mathrm{train}}(A)
    \right|
    \leq \varepsilon_{\mathrm{cal}}
\right] \geq 1 - \delta_{\mathrm{cal}}.
\]
\end{assumption}

\begin{assumption}[A5: Conditional test decoupling]\label{assump:A5_app}
There exists $\varepsilon_{\mathrm{train}} \in (0,1)$ such that, for all $A \in \mathcal{A}$, $i \in I_{\mathrm{test}}$ and every $q_{\mathrm{train}}$:
\[
\big|
\mathbb{P}_{\mathrm{tr}}[\hat{s}_{\mathrm{train}}(X_i,Y_i) \leq q_{\mathrm{train}} \mid X_i \in A]
- \mathbb{P}_{\mathrm{tr}}[\hat{s}_{\mathrm{train}}(X_*,Y_*) \leq q_{\mathrm{train}} \mid X_* \in A]
\big|
\leq \varepsilon_{\mathrm{train}}.
\]
\end{assumption}

\subsection{Coverage Guarantees and Proofs}

\begin{theorem}[Marginal coverage under non-exchangeability]
\label{theo 1_app}
Let $\alpha \in (0,1)$, and suppose conditions \hyperref[assump:A1_app]{(A1)} and \hyperref[assump:A2_app]{(A2)} hold. Then, for all $i \in I_{\mathrm{test}}$,
\begin{equation}
    \mathbb{P}_{\mathrm{tr}} \left[ Y_i \in \mathcal{C}_{1 - \alpha}(X_i) \right] 
    \geq 1 - \alpha - \varepsilon_{\mathrm{cal}} - \delta_{\mathrm{cal}} - \varepsilon_{\mathrm{train}}.
    \label{eq:marginal-coverage_app}
\end{equation}
Additionally, if $\hat{s}_{\mathrm{train}}(X_*, Y_*)$ has a continuous distribution almost surely, conditional on the training data, then:
\begin{equation}
    \left| \mathbb{P}_{\mathrm{tr}} \left[ Y_i \in \mathcal{C}_{1 - \alpha}(X_i) \right] - (1 - \alpha) \right| 
    \leq \varepsilon_{\mathrm{cal}} + \delta_{\mathrm{cal}} + \varepsilon_{\mathrm{train}}.
\end{equation}
\end{theorem}

\begin{proof}
We define the event
\[
F := \left\{ \hat q_{1-\alpha,\mathrm{cal}} \ \ge \ q_{1-\alpha-\varepsilon_{\mathrm{cal}},\mathrm{train}}  \right\}.
\]

\paragraph{Step 1 (Bounding $\mathbb{P}_{\mathrm{tr}}(F)$)}
By \hyperref[assump:A1_app]{(A1)}, for any $\ell \in \mathbb{N}_{>0}$, it holds with probability at least $1-\delta_{\mathrm{cal}}$ that
\[
\begin{aligned}
\frac{1}{n_{\mathrm{cal}}} \sum_{i \in I_{\mathrm{cal}}} 
\mathbf{1} \left\{ \hat{s}_{\mathrm{train}}(X_i, Y_i) 
\leq q_{1-\alpha-\varepsilon_{\mathrm{cal}},\mathrm{train}} - \frac{1}{\ell} \right\} 
&\leq \mathbb{P}_{\mathrm{tr}} \left[ \hat{s}_{\mathrm{train}}(X_*, Y_*) 
\leq q_{1-\alpha-\varepsilon_{\mathrm{cal}},\mathrm{train}} - \frac{1}{\ell} \right] 
+ \varepsilon_{\mathrm{cal}} \\
&< 1 - \alpha - \varepsilon_{\mathrm{cal}} + \varepsilon_{\mathrm{cal}} \\
&= 1 - \alpha \\
&\leq \frac{1}{n_{\mathrm{cal}}} \sum_{i \in I_{\mathrm{cal}}} 
\mathbf{1} \left\{ \hat{s}_{\mathrm{train}}(X_i, Y_i) 
\leq \hat{q}_{1-\alpha,\mathrm{cal}} \right\}.
\end{aligned}
\]
The second and last inequality holds by the definition of the quantiles.

Now define
\[
E_\ell := \left\{
\frac{1}{n_{\mathrm{cal}}} \sum_{i \in I_{\mathrm{cal}}}
\mathbf{1} \left\{ \hat{s}_{\mathrm{train}}(X_i, Y_i) \leq q_{1 - \alpha - \varepsilon_\mathrm{cal}, \mathrm{train}} - \frac{1}{\ell} \right\}
<
\frac{1}{n_{\mathrm{cal}}} \sum_{i \in I_{\mathrm{cal}}}
\mathbf{1} \left\{ \hat{s}_{\mathrm{train}}(X_i, Y_i) \leq \hat{q}_{1 - \alpha , \mathrm{cal}} \right\}
\right\}.
\]
As just shown, $\mathbb{P}_{\mathrm{tr}}[E_\ell] \geq 1 - \delta_{\mathrm{cal}}$.
By continuity from above for probabilities, $\mathbb{P}_{\mathrm{tr}} ( \bigcap_{\ell=1}^\infty E_\ell ) = \lim_{\ell \to \infty} \mathbb{P}_{\mathrm{tr}}(E_\ell) \ge 1 - \delta_{\mathrm{cal}}$.
On the intersection $\bigcap_{\ell=1}^\infty E_\ell$, we have $\hat q_{1-\alpha,\mathrm{cal}} \ge q_{1-\alpha-\varepsilon_{\mathrm{cal}},\mathrm{train}} - \frac{1}{\ell}$ for all $\ell$, which implies $\hat q_{1-\alpha,\mathrm{cal}} \ge q_{1-\alpha-\varepsilon_{\mathrm{cal}},\mathrm{train}}$.
Thus, $\mathbb{P}_{\mathrm{tr}}(F) \ge 1 - \delta_{\mathrm{cal}}$.

\paragraph{Step 2 (First Bound).}
Fix $i \in I_{\mathrm{test}}$. Since $t \mapsto \mathbb{P}_{\mathrm{tr}}\left[\hat s_{\mathrm{train}}(X_i,Y_i) \le t\right]$ is nondecreasing,
on the event $F$ we have $\{\hat s_{\mathrm{train}}(X_i,Y_i) \le q_{1-\alpha-\varepsilon_{\mathrm{cal}},\mathrm{train}}\} \subseteq \{\hat s_{\mathrm{train}}(X_i,Y_i) \le \hat q_{1-\alpha,\mathrm{cal}}\}$.
Therefore,
\begin{align*}
\mathbb{P}_{\mathrm{tr}} \left[\hat s_{\mathrm{train}}(X_i,Y_i) \le \hat q_{1-\alpha,\mathrm{cal}}\right]
&\ge \mathbb{P}_{\mathrm{tr}} \left( \{\hat s_{\mathrm{train}}(X_i,Y_i) \le \hat q_{1-\alpha,\mathrm{cal}}\} \cap F \right) \\
&\ge \mathbb{P}_{\mathrm{tr}} \left( \{\hat s_{\mathrm{train}}(X_i,Y_i) \le q_{1-\alpha-\varepsilon_{\mathrm{cal}},\mathrm{train}}\} \cap F \right) \\
&\ge \mathbb{P}_{\mathrm{tr}} \left[\hat s_{\mathrm{train}}(X_i,Y_i) \le q_{1-\alpha-\varepsilon_{\mathrm{cal}},\mathrm{train}}\right] - \mathbb{P}_{\mathrm{tr}}(F^c).
\end{align*}
Since $\mathbb{P}_{\mathrm{tr}}(F^c) \le \delta_{\mathrm{cal}}$, we conclude that
\[
\mathbb{P}_{\mathrm{tr}} \left[\hat s_{\mathrm{train}}(X_i,Y_i) \le \hat q_{1-\alpha,\mathrm{cal}}\right]
\ \ge\ 
\mathbb{P}_{\mathrm{tr}} \left[\hat s_{\mathrm{train}}(X_i,Y_i) \le q_{1-\alpha-\varepsilon_{\mathrm{cal}},\mathrm{train}}\right] - \delta_{\mathrm{cal}}.
\]
By \hyperref[assump:A2_app]{(A2)} we have
\[
\mathbb{P}_{\mathrm{tr}} \left[\hat s_{\mathrm{train}}(X_i,Y_i) \le q_{1-\alpha-\varepsilon_{\mathrm{cal}},\mathrm{train}}\right]
\ \ge\
\mathbb{P}_{\mathrm{tr}} \left[\hat s_{\mathrm{train}}(X_*,Y_*) \le q_{1-\alpha-\varepsilon_{\mathrm{cal}},\mathrm{train}}\right] 
- \varepsilon_{\mathrm{train}}.
\]
Using the definition of the quantile, we have
\[
\mathbb{P}_{\mathrm{tr}} \left[\hat s_{\mathrm{train}}(X_i,Y_i) \le \hat q_{1-\alpha,\mathrm{cal}}\right]
\ \ge\
\mathbb{P}_{\mathrm{tr}} \left[\hat s_{\mathrm{train}}(X_*,Y_*) \le q_{1-\alpha-\varepsilon_{\mathrm{cal}},\mathrm{train}}\right] - \varepsilon_{\mathrm{train}} - \delta_{\mathrm{cal}}
\ \ge\
1 - \alpha - \varepsilon_{\mathrm{cal}} - \varepsilon_{\mathrm{train}} - \delta_{\mathrm{cal}}.
\]
This proves the first bound.

\paragraph{Step 3 (Second Bound).}
Assume that $\hat{s}_{\mathrm{train}}(X_*,Y_*)$ has a continuous distribution almost surely, conditional on the training data.
Define $G := \{\, \hat q_{1-\alpha,\mathrm{cal}} \le q_{1-\alpha+\varepsilon_{\mathrm{cal}},\mathrm{train}} \,\}$.
Proceeding exactly as in Step~1, but using the lower-tail direction of \hyperref[assump:A1_app]{(A1)}, we obtain $\mathbb{P}_{\mathrm{tr}}(G) \ge 1-\delta_{\mathrm{cal}}$.
Now let $A := \{\hat s_{\mathrm{train}}(X_i,Y_i) \le \hat q_{1-\alpha,\mathrm{cal}}\}$.
\begin{align*}
\mathbb{P}_{\mathrm{tr}}(A)
&= \mathbb{P}_{\mathrm{tr}}\big(A \cap G\big) + \mathbb{P}_{\mathrm{tr}}\big(A \cap G^c\big)\\
&\le \mathbb{P}_{\mathrm{tr}}\big(\{\hat s_{\mathrm{train}}(X_i,Y_i) \le \hat q_{1-\alpha,\mathrm{cal}}\} \cap G\big) + \mathbb{P}_{\mathrm{tr}}(G^c)\\
&\le \mathbb{P}_{\mathrm{tr}}\big[\hat s_{\mathrm{train}}(X_i,Y_i) \le q_{1-\alpha+\varepsilon_{\mathrm{cal}},\mathrm{train}}\big] + \mathbb{P}_{\mathrm{tr}}(G^c)\\
&\le \mathbb{P}_{\mathrm{tr}}\big[\hat s_{\mathrm{train}}(X_*,Y_*) \le q_{1-\alpha+\varepsilon_{\mathrm{cal}},\mathrm{train}}\big] + \varepsilon_{\mathrm{train}} + \mathbb{P}_{\mathrm{tr}}(G^c) \quad (\text{by (A2)}) \\
&\le \mathbb{P}_{\mathrm{tr}}\big[\hat s_{\mathrm{train}}(X_*,Y_*) \le q_{1-\alpha+\varepsilon_{\mathrm{cal}},\mathrm{train}}\big] + \varepsilon_{\mathrm{train}} + \delta_{\mathrm{cal}} \\
&= 1-\alpha+\varepsilon_{\mathrm{cal}} + \varepsilon_{\mathrm{train}} + \delta_{\mathrm{cal}}. \quad (\text{by continuity})
\end{align*}
Combining this with the first bound yields the second bound.
\end{proof}

\begin{theorem}[Empirical coverage under non-exchangeability]
\label{theo 2_app}
Let $\alpha \in (0,1)$, \(\delta_{\mathrm{cal}} > 0\) and \(\delta_{\mathrm{test}} > 0\), if \hyperref[assump:A1_app]{(A1)} and \hyperref[assump:A3_app]{(A3)} hold, then: 
\begin{equation}
    \mathbb{P}_{\mathrm{tr}}\left[ \frac{1}{n_{\mathrm{test}}} \sum_{i \in I_{\mathrm{test}}} \mathbf{1}\{Y_i \in \mathcal{C}_{1-\alpha}(X_i)\} \geq 1-\alpha - \eta\right]
    \ \ge\ 1 - \delta_{\mathrm{cal}} - \delta_{\mathrm{test}}.
\end{equation}
where \( \eta = \varepsilon_{\mathrm{cal}} + \varepsilon_{\mathrm{test}}. \)
Additionally, if $\hat{s}_{\mathrm{train}}(X_*, Y_*)$ almost surely has a continuous distribution conditionally on the training data, then the bound can be tightened to:
\begin{equation}
    \mathbb{P}_{\mathrm{tr}}\left[\left| \frac{1}{n_{\mathrm{test}}} \sum_{i \in I_{\mathrm{test}}} \mathbf{1}\{Y_i \in \mathcal{C}_{1-\alpha}(X_i)\} - (1-\alpha) \right|\leq \eta\right]
    \ \ge\ 1- 2\delta_{\mathrm{cal}}-2\delta_{\mathrm{test}}.
\end{equation}
\end{theorem}

\begin{proof}
Using an argument analogous to that in the proof of Theorem~\ref{theo 1_app}, we can show that the event
\[
F := \left\{\, \hat{q}_{1-\alpha-\eta,\mathrm{test}} \ \le\ \hat{q}_{1-\alpha,\mathrm{cal}} \,\right\}
\]
satisfies
\[
\mathbb{P}_{\mathrm{tr}}[F] \ \ge\ 1 - \delta_{\mathrm{cal}} - \delta_{\mathrm{test}}.
\]
Therefore,
\begin{align*}
&\mathbb{P}_{\mathrm{tr}}\!\left[
\frac{1}{n_{\mathrm{test}}} \sum_{i\in I_{\mathrm{test}}}
\mathbf{1}\{ \hat{s}_{\mathrm{train}}(X_i,Y_i) \le \hat{q}_{1-\alpha,\mathrm{cal}} \}
\ \ge\ 1 - \alpha - \eta
\right] \\
&\quad\ge\
\mathbb{P}_{\mathrm{tr}}\!\left[ \{
\dots \} \cap F \right] - \mathbb{P}_{\mathrm{tr}}(F^c) \\
&\quad\geq\
\mathbb{P}_{\mathrm{tr}}\!\left[
\frac{1}{n_{\mathrm{test}}} \sum_{i\in I_{\mathrm{test}}}
\mathbf{1}\{ \hat{s}_{\mathrm{train}}(X_i,Y_i) \le \hat{q}_{1-\alpha-\eta,\mathrm{test}} \}
\ \ge\ 1 - \alpha - \eta
\right]
- \delta_{\mathrm{cal}} - \delta_{\mathrm{test}} \\
&\quad\ge\ 1 - \delta_{\mathrm{cal}} - \delta_{\mathrm{test}},
\end{align*}
which establishes the first claim.
For the second claim, define $G := \left\{\, \hat{q}_{1-\alpha,\mathrm{cal}} \ \leq \ \hat{q}_{1-\alpha+\eta,\mathrm{test}} \,\right\}$.
By the same reasoning, $\mathbb{P}_{\mathrm{tr}}(G) \ge 1 - \delta_{\mathrm{cal}} - \delta_{\mathrm{test}}$.
The event $F \cap G$ has probability $\mathbb{P}_{\mathrm{tr}}(F \cap G) \ge 1 - 2\delta_{\mathrm{cal}} - 2\delta_{\mathrm{test}}$.
On this event, $ \hat{q}_{1-\alpha-\eta,\mathrm{test}} \le \hat{q}_{1-\alpha,\mathrm{cal}} \le \hat{q}_{1-\alpha+\eta,\mathrm{test}} $.
The continuity assumption ensures the quantiles correspond to the desired probabilities, and the result follows.
\end{proof}

\begin{theorem}[Conditional coverage under non-exchangeability]
\label{theo 3_app}
Let $\alpha \in (0,1)$ and $\delta_{\mathrm{cal}} > 0$.  
If \hyperref[assump:A4_app]{(A4)} and \hyperref[assump:A5_app]{(A5)} hold, then for every $A \in \mathcal{A} \subset \mathcal{X}$ and any $i \in I_{\mathrm{test}}$,
\begin{equation}
\mathbb{P}_{\mathrm{tr}}\big[ Y_i \in \mathcal{C}_{1-\alpha}(X_i;A) \mid X_i \in A \big] 
\geq 1 - \alpha - \varepsilon_{\mathrm{cal}} - \delta_{\mathrm{cal}} - \varepsilon_{\mathrm{train}}.
\end{equation}
Furthermore, if $\hat{s}_{\mathrm{train}}(X_*,Y_*)$ almost surely has a continuous distribution given the training data, then
\begin{equation}
\big| \mathbb{P}_{\mathrm{tr}}[Y_i \in \mathcal{C}_{1-\alpha}(X_i;A) \mid X_i \in A] - (1-\alpha) \big| 
\leq \varepsilon_{\mathrm{cal}} + \delta_{\mathrm{cal}} + \varepsilon_{\mathrm{train}}.
\end{equation}
\end{theorem}
\begin{proof}
As in Step 1 of the proof of Theorem~\ref{theo 1_app}, but using the uniform assumption \hyperref[assump:A4_app]{(A4)}, with probability at least $1 - \delta_{\mathrm{cal}}$ the event
\[
F_{\mathrm{cal}} := \big\{\, q_{1-\alpha-\varepsilon_{\mathrm{cal}}}(A) \leq \hat{q}_{1-\alpha,\mathrm{cal}}(A),\ \forall A \in \mathcal{A} \,\big\}
\]
occurs.  
Following Step 2, we have
\[
\mathbb{P}_{\mathrm{tr}}[\hat{s}_{\mathrm{train}}(X_i,Y_i) \leq \hat{q}_{1-\alpha,\mathrm{cal}}(A) \mid X_i \in A]
\geq \mathbb{P}_{\mathrm{tr}}[\hat{s}_{\mathrm{train}}(X_i,Y_i) \leq q_{1-\alpha-\varepsilon_{\mathrm{cal}}}(A) \mid X_i \in A] - \delta_{\mathrm{cal}}.
\]
Applying \hyperref[assump:A5_app]{(A5)} yields:
\begin{align*}
&\mathbb{P}_{\mathrm{tr}}\big[\hat{s}_{\mathrm{train}}(X_i,Y_i) \leq \hat{q}_{1-\alpha,\mathrm{cal}}(A) \,\big|\, X_i \in A\big] \\
&\quad\geq\ \mathbb{P}_{\mathrm{tr}}\big[\hat{s}_{\mathrm{train}}(X_*,Y_*) \leq q_{1-\alpha-\varepsilon_{\mathrm{cal}}}(A) \,\big|\, X_* \in A\big] 
    - \varepsilon_{\mathrm{train}} - \delta_{\mathrm{cal}} \\
&\quad\geq\ 1 - \alpha - \varepsilon_{\mathrm{cal}} - \varepsilon_{\mathrm{train}} - \delta_{\mathrm{cal}}.
\end{align*}
The bound under the continuity assumption follows by repeating the same argument as in Theorem~\ref{theo 1_app}.
\end{proof}

\subsection{Guarantees for $\beta$-Mixing Processes}
\label{sec:beta-mixing_app}
We now show how to derive explicit slack terms $(\delta, \varepsilon)$ for stationary \emph{$\beta$-mixing} (absolutely regular) processes.

\begin{definition}[$\beta$-mixing]
Let $(Z_t)_{t\ge1}$ be a sequence. The $\beta$-mixing coefficient at lag $a\in\mathbb{N}$ is
\[
\beta(a)
:= \big\| \mathbb{P}_{-\infty:0,\,a:\infty}
         - \mathbb{P}_{-\infty:0}\otimes \mathbb{P}_{a:\infty} \big\|_{\mathrm{TV}},
\]
where $\|\cdot\|_{\mathrm{TV}}$ is the total variation distance between the joint law of the past and future, and the product of their marginals. The process is $\beta$-mixing if $\beta(a)\to 0$ as $a\to\infty$.
\end{definition}

This "forgetting" property allows us to use a \emph{blocking technique}.

\begin{proposition}[Blocking Technique]
\label{prop:blocking-technique}
Let $Z_{t=1}^{T}$ be a sample of a stationary $\beta$-mixing process, split into $2m$ interleaved blocks (odd blocks of size $b$, even blocks of size $a$). Let $B_{\mathrm{odd}} = (B_1, B_3, \dots, B_{2m-1})$ be the set of odd blocks, and $B_{\mathrm{odd}}'$ be an independent version.
If $h : \mathbb{R}^{mb} \to \mathbb{R}$ is a Borel-measurable function with $|h| \le M$, then
\[
\big| \mathbb{E}[h(B_{\mathrm{odd}})] - \mathbb{E}[h(B_{\mathrm{odd}}')] \big| \ \le\ 2M (m-1) \beta(a),
\]
where $\beta(a)$ is the $\beta$-mixing coefficient.
\end{proposition}

We also use two standard results:

\begin{lemma}
\label{lem:mohri-rostamizadeh}
Let $Z_{1:n}$ be a sample from a stationary $\beta$-mixing distribution and $\mathcal{F}$ be a class of functions from $\mathcal{X}$ to $[0,1]$. Split the sample into $2m$ blocks of size $a$ ($n = 2ma$). Let $B_{\mathrm{odd}}$ be the odd blocks and $B_{\mathrm{odd}}'$ their independent version. Then,
\[
\mathbb{P} \left( \sup_{f \in \mathcal{F}} \left\lvert\mathbb{E}[f(Z_1)] - \frac{1}{n} \sum_{i=1}^n f(Z_i)\right\rvert > \varepsilon  \right)
\ \le\ 
2\,\mathbb{P}' \left( \sup_{f \in \mathcal{F}} \left\lvert\mathbb{E}[f(Z_1)] - \frac{1}{ma} \sum_{Z_j \in B_{\mathrm{odd}}'} f(Z_j)\right\rvert > \varepsilon \right)
+ 4(m-1)\,\beta(a).
\]
\end{lemma}

\begin{lemma}[Bernstein’s inequality, \cite{vershynin_high-dimensional_2018}]
\label{lem:bernstein}
Let $X_1, \dots, X_m \in [0,1]$ be independent random variables with 
$\mathrm{Var}(X_j) \le \sigma^2$. Then, for any $\delta \in (0,1)$, with probability at least $1 - \delta$,
\[
\left|\mathbb{E}[X_j] - \frac{1}{m} \sum_{j=1}^m X_j\right|
\ \le\ \sigma
\sqrt{\frac{2\log(1/\delta)}{m}} 
+ \frac{\log(1/\delta)}{3m}.
\]
\end{lemma}

These tools lead to the key technical lemma:

\begin{lemma}
\label{lem:lemma12}
Let $Z_{1:n}$ be a sample drawn from a stationary $\beta$-mixing distribution with $Z_1 \in [0,1]$ and $\mathrm{Var}[Z_1] = v < \infty$. Then, for $n = 2ma + s$ and $\delta > 4(m-1)\beta(a)$, with probability at least $1 - \delta$ it holds that
\[
\left\lvert\mathbb{E}[Z_1] - \frac{1}{n} \sum_{i=1}^n Z_i \right\rvert \ \le\ \varepsilon 
\]
where
\[
\varepsilon = \tilde{\sigma}(a) \sqrt{\frac{4}{n}\log \left({\frac{4}{\delta-4(m-1)\beta(a)}} \right)}
+ \frac{1}{3m}\log\left({\frac{4}{\delta-4(m-1)\beta(a)}} \right)
+ \frac{s}{n}
\]
and $\tilde{\sigma}(a) = \sqrt{v + \frac{2}{a} \sum_{k=1}^{a-1} (a - k) \beta(k)}$.
\end{lemma}

\begin{proof}
\noindent\textbf{Step 1 (Reduce the full mean).}
$\bigl|\mathbb{E}[Z_1]-\frac{1}{n}\sum_{i=1}^n Z_i\bigr| \le \frac{2ma}{n}\bigl|\mathbb{E}[Z_1]-\frac{1}{2ma}\sum_{i=1}^{2ma} Z_i\bigr| + \frac{s}{n}$.
\medskip
\noindent\textbf{Step 2 (Apply Lemma~\ref{lem:mohri-rostamizadeh} with $f(x)=x$).}
Let $\bar Z$ be the mean of the first $2ma$ points and $\bar Z'$ be the mean of the odd blocks in an independent version. For $\delta' = \delta - 4(m-1)\beta(a)$,
\[
\mathbb{P}\!\left(\left|\mathbb{E}[Z_1] - \bar Z \right|> \varepsilon\right)
\ \le\ 
2\,\mathbb{P}'\!\left(\left|\mathbb{E}[Z_1] - \bar Z'\right| > \varepsilon\right) + \delta - \delta'.
\]
\medskip
\noindent\textbf{Step 3 (Bernstein on the independent block averages).}
Apply Lemma~\ref{lem:bernstein} to $\bar Z'$ (which is an average of $m$ independent block averages $X_j$) with confidence level $\delta'/2$.
Let $\sigma_a^2 = \mathrm{Var}(X_j)$. With probability at least $1-\delta'/2$ under $\mathbb{P}'$,
\[
\left|\mathbb{E}[Z_1] - \bar Z'\right|
\ \le\ \sigma_a
\sqrt{\frac{2\,\log\!\big(4/\delta'\big)}{m}} + \frac{\log\!\big(4/\delta'\big)}{3m}.
\]
\medskip
\noindent\textbf{Step 4 (Identify $\tilde{\sigma}(a)$).}
We bound $\sigma_a^2 = \mathrm{Var}(\frac{1}{a}\sum_{i=1}^a Z_i)$ using $|\mathrm{Cov}(Z_1,Z_{1+k})| \le \beta(k)$ \cite{doukhan_mixing_1994}:
\[
\sigma_a^2 \le \frac{1}{a^2} \left[ a v + 2 \sum_{k=1}^{a-1} (a-k) \beta(k) \right] = \frac{1}{a} \tilde{\sigma}(a)^2.
\]
Substituting $\sigma_a = \tilde{\sigma}(a)/\sqrt{a}$ and $2ma \approx n$ gives the result.
\end{proof}

This lemma is the engine for verifying assumptions (A1)-(A3).

\begin{proposition}
\label{prop: 13_app}
(Verifying A1) Assume that $(X_i, Y_i)_{i=1}^n$ is a stationary $\beta$-mixing process. Then \hyperref[assump:A1_app]{(A1)} is satisfied with
\begin{equation}\label{epscal_app}
\begin{aligned}
\varepsilon_{\mathrm{cal}}
  &= \inf_{(a,m,r)\in\mathcal{F}_{\mathrm{cal}}} \Bigg\{ \;
      \tilde{\sigma}(a)\,
      \sqrt{\frac{4}{\,n_{\mathrm{cal}}-r+1\,}
      \log\!\left(
        \frac{4}{\delta_{\mathrm{cal}}
          -4(m-1)\beta(a)-\beta(r)}
      \right)} \\[6pt]
  &\qquad\qquad
      + \frac{1}{3m}\,
        \log\!\left( \dots \right)
      + \frac{r-1}{n_{\mathrm{cal}}}
    \;\Bigg\}.
\end{aligned}
\end{equation}
where $\mathcal{F}_{\mathrm{cal}} = \bigl\{ (a,m,r)\in\mathbb{N}_{>0}^3 : 2ma = n_{\mathrm{cal}}-r+1,\; \delta_{\mathrm{cal}} > 4(m-1)\beta(a)+\beta(r) \bigr\}$
and $\tilde{\sigma}(a) = \sqrt{\frac{1}{4} + \frac{2}{a}\sum_{k=1}^{a-1}(a-k)\beta(k)}$ (using $v \le 1/4$ for indicator variables).
\end{proposition}
\begin{proof}
We want to apply Lemma~\ref{lem:lemma12} to $Z_i = \mathbf{1}\{\hat{s}_{\mathrm{train}}(X_i,Y_i)\le q_{\mathrm{train}}\}$ for $i \in I_{\mathrm{cal}}$. These $Z_i$ are not independent of $q_{\mathrm{train}}$ (which depends on $I_{\mathrm{train}}$).
\textbf{Step 1. Decoupling via a gap.}
Define a shifted calibration set $I_{\mathrm{cal},r}=\{n_{\mathrm{train}}+r,\ldots,n_{\mathrm{train}}+n_{\mathrm{cal}}\}$ and the event $E(r,\varepsilon)$ of deviation on this set. Let $\mathbb{P}_*$ be the product measure where $I_{\mathrm{train}}$ and $I_{\mathrm{cal},r}$ are independent. The total variation distance between the true law and $\mathbb{P}_*$ is at most $\beta(r)$.
Thus, $\mathbb{P}_{\mathrm{tr}}[E(1,\varepsilon)] \le \mathbb{P}_{\mathrm{tr}}[E(r,\varepsilon-\frac{r-1}{n_{\mathrm{cal}}})] \le \mathbb{P}_*[E(r,\varepsilon-\frac{r-1}{n_{\mathrm{cal}}})] + \beta(r)$.
\textbf{Step 2. Concentration under independence.}
Working under $\mathbb{P}_*$, $q_{\mathrm{train}}$ is fixed, and the $Z_i$ for $i \in I_{\mathrm{cal},r}$ form a $\beta$-mixing sequence of length $n_{\mathrm{cal}}-r+1$. We apply Lemma~\ref{lem:lemma12} to this sequence.
\textbf{Step 3. Optimizing the bound.}
This gives a probability bound $\delta_{\mathrm{cal}} - 4(m-1)\beta(a) - \beta(r)$ for a deviation of size $\varepsilon'$. We choose $\varepsilon_{\mathrm{cal}}$ to be the infimum of this $\varepsilon'$ over all valid choices of $(a,m,r)$.
\end{proof}

\begin{proposition}
\label{prop: 14_app}
(Verifying A2) If $(X_i,Y_i)_{i=1}^n$ is a stationary $\beta$-mixing process, then \hyperref[assump:A2_app]{(A2)} holds with
\[
\varepsilon_{\mathrm{train}} \;=\; \beta(k-n_{\mathrm{train}}) =: \beta_k, \quad k \in I_\mathrm{test}
\]
\end{proposition}
\begin{proof}
Fix $k \in I_\mathrm{test}$ and consider the product measure $\mathbb{P}_* = \mathbb{P}_{1}^{n_{\mathrm{train}}} \otimes \mathbb{P}_k^{k}$, where the $k$-th point is independent of training. By $\beta$-mixing (with gap $k-n_{\mathrm{train}}$),
\begin{align*}
\beta_k &\;\geq\;
 \left\vert\mathbb{P}_{\mathrm{tr}}\left[
   \hat{s}_{\mathrm{train}}(X_k,Y_k) \le q_{\mathrm{train}}\right]
   -\mathbb{P}_*\left[
   \hat{s}_{\mathrm{train}}(X_k,Y_k) \le q_{\mathrm{train}}
\right] \right\vert \\
&= 
\left\vert \mathbb{P}_{\mathrm{tr}}\left[
   \hat{s}_{\mathrm{train}}(X_k,Y_k) \le q_{\mathrm{train}}\right]  - 
   \mathbb{E}_*\left[
   \mathbb{P}_*\left[
     \hat{s}_{\mathrm{train}}(X_k,Y_k) \le q_{\mathrm{train}}
     \,\mid\,\mathcal{F}_{\mathrm{train}}
   \right]\right]\right\vert  
 \\
&= 
\left\vert \mathbb{P}_{\mathrm{tr}}\left[
   \hat{s}_{\mathrm{train}}(X_k,Y_k) \le q_{\mathrm{train}}\right]  - 
   \mathbb{P}_{\mathrm{tr}}\left[
     \hat{s}_{\mathrm{train}}(X_*,Y_*) \le q_{\mathrm{train}}
   \right]\right\vert
\end{align*}
This is exactly the bound required by \hyperref[assump:A2_app]{(A2)}.
\end{proof}

\begin{proposition}
\label{prop: 15_app}
(Verifying A3) Assume that $(X_i, Y_i)_{i=1}^n$ is a stationary $\beta$-mixing process. Then \hyperref[assump:A3_app]{(A3)} is satisfied with
\begin{equation} \label{epstest_app}
\begin{aligned}
\varepsilon_{\mathrm{test}}
  = \inf_{(a,m,s)\in \mathcal{F}_{\mathrm{test}}}
      \Biggl\{&
        \tilde{\sigma}(a)\,
        \sqrt{\frac{4}{n_{\mathrm{test}}}
           \log\!\left(
             \frac{4}{\delta_{\mathrm{test}}
                   -4(m-1)\beta(a)-\beta(n_{\mathrm{cal}})}
           \right)} \\[6pt]
       &+ \frac{1}{3m}\,
           \log\!\left( \dots \right)
        + \frac{s}{n_{\mathrm{test}}}
      \Biggr\},
\end{aligned}
\end{equation}
where $\mathcal{F}_{\mathrm{test}} = \Bigl\{(a,m,s)\in\mathbb{N}^2\times\mathbb{N}_{\ge0}:\, s+2ma=n_{\mathrm{test}},\;\; \delta_{\mathrm{test}}>4(m-1)\beta(a)+\beta(n_{\mathrm{cal}}) \Bigr\}$.
\end{proposition}
\begin{proof}
The argument parallels the proof of Proposition~\ref{prop: 13_app}.
Define the event $E(\varepsilon)$ for the test set. Let $\mathbb{P}_* = \mathbb{P}_{1}^{\,n_{\mathrm{train}}} \otimes \mathbb{P}_{n_{\mathrm{train}}+n_{\mathrm{cal}}+1}^{\,n_{\mathrm{train}}+n_{\mathrm{cal}}+n_{\mathrm{test}}}$ be the product measure where $I_{\mathrm{train}}$ and $I_{\mathrm{test}}$ are independent. The gap is $n_{\mathrm{cal}}$.
By $\beta$-mixing, $\mathbb{P}_{\mathrm{tr}}[E(\varepsilon)] \le \mathbb{P}_*[E(\varepsilon)] + \beta(n_{\mathrm{cal}})$.
Applying Lemma~\ref{lem:lemma12} to the $n_{\mathrm{test}}$ points under $\mathbb{P}_*$ (which are $\beta$-mixing) gives the concentration bound. Optimizing over $(a,m,s)$ yields the result.
\end{proof}

These propositions lead directly to concrete versions of the main theorems:

\begin{theorem}[Marginal coverage: $\beta$-mixing]
\label{theo 4_app}
Suppose the sample $(X_i,Y_i)_{i=1}^n$ is stationary and $\beta$-mixing.
Then for $i\in I_{\mathrm{test}}$ we have
\[
\mathbb{P}_{\mathrm{tr}}[\,Y_i \in \mathcal{C}_{1-\alpha}(X_i)\,] \;\ge\; 1-\alpha-\eta, \quad \text{with} \quad \eta \;=\; \delta_{\mathrm{cal}} + \varepsilon_{\mathrm{train}} + \varepsilon_{\mathrm{cal}},
\]
where $\varepsilon_{\mathrm{cal}}$ is from \ref{epscal_app} and
$\varepsilon_{\mathrm{train}}=\beta(i-n_{\mathrm{train}})$.
\end{theorem}
\begin{proof}
This follows by substituting the bounds from Propositions~\ref{prop: 13_app} and~\ref{prop: 14_app} into the general result of Theorem~\ref{theo 1_app}.
\end{proof}

\begin{theorem}[Empirical coverage: $\beta$-mixing]
\label{theo 5_app}
Suppose the sample $(X_i,Y_i)_{i=1}^n$ is stationary and $\beta$-mixing.
Then
\[
\mathbb{P}_{\mathrm{tr}}\!\left[
  \frac{1}{n_{\mathrm{test}}}\sum_{i\in I_{\mathrm{test}}}
    1\{Y_i \in \mathcal{C}_{1-\alpha}(X_i)\} \;\ge\; 1-\alpha-\eta
\right] \;\ge\; 1-\delta_{\mathrm{cal}}-\delta_{\mathrm{test}},
\]
with $\eta \;=\; \varepsilon_{\mathrm{cal}} + \varepsilon_{\mathrm{test}}$, where $\varepsilon_{\mathrm{cal}}$ and $\varepsilon_{\mathrm{test}}$
are defined in \ref{epscal_app} and \ref{epstest_app}.
\end{theorem}
\begin{proof}
This follows by substituting the bounds from Propositions~\ref{prop: 13_app} and~\ref{prop: 15_app} into the general result of Theorem~\ref{theo 2_app}.
\end{proof}
 
\begin{theorem}[(Conditional coverage: $\beta$-mixing)]
\label{theo 16_app}
Suppose that $(X_i,Y_i)_{i=1}^n$ is stationary $\beta$-mixing. Then given
$\alpha\in(0,1)$, $\gamma>0$ and $\delta_{\mathrm{cal}}>0$, for each $A\in\mathcal{A}$
and any $i\in I_{\mathrm{test}}$,
\[
\mathbb{P}_{\mathrm{tr}}\!\left[\,Y_i \in \mathcal{C}_{1-\alpha}(X_i;A)\mid X_i\in A\,\right]\;\ge\;1-\alpha-\eta,
\]
with $\eta=\varepsilon_{\mathrm{cal}}+\varepsilon_{\mathrm{train}}$, where
$\varepsilon_{\mathrm{cal}}$ is as in \ref{epscal_app} and
$\varepsilon_{\mathrm{train}}=\beta(i-n_{\mathrm{train}})$.
\end{theorem}
\begin{proof}
The proof follows the same logic as Theorem \ref{theo 4_app}, but uses the conditional guarantee from Theorem \ref{theo 3_app} and requires verifying the conditional assumptions \hyperref[assump:A4_app]{(A4)} and \hyperref[assump:A5_app]{(A5)} using the same blocking and decoupling techniques.
\end{proof}

\section{Specific Theoretical Guarantees}
\label{app:bounds}

This appendix contains the specific theoretical guarantees for the methods discussed in Sections \ref{sec:wcp},\ref{sec:enbpi}, as referenced in the main text.

\begin{theorem}[Nex-CP marginal coverage bound, Theorem 2/3 in \cite{barber_conformal_2023}]
\label{thm:wcp_app}
Let $\boldsymbol{s}$ be the vector of scores (calibration + test point $m$). For any \(\alpha\in(0,1)\), the WCP method satisfies
\[
\Bigl|\;\mathbb{P}_{\mathrm{tr}}\!\left[\,Y_{m}\in \mathcal{C}^{(w)}_{1-\alpha}(X_{m})\,\right] - (1-\alpha)\Bigr|
\ \le\ \sum_{i\in I_{\mathrm{cal}}} \widetilde w_i \, d_{\mathrm{TV}}\!\bigl(\mathcal{L}(\boldsymbol{s}),\,\mathcal{L}(\boldsymbol{s}^{(i)})\bigr).
\]
where $\boldsymbol{s}^{(i)}$ is the score vector with test point $m$ and calibration point $i$ swapped.
\end{theorem}

\begin{theorem}[CRC bound, Theorem 1 in \cite{farinhas_non-exchangeable_2024}]
\label{thm:crc_app}
Suppose the loss function $\ell$ is bounded in $[A,B]$. Then the weighted CRC selector $\hat{\lambda}$ satisfies
\[
\Bigl|\,\mathbb{E}\!\left[L_m(\hat{\lambda})\right] - \alpha \Bigr|
\ \le\ (B-A)\sum_{i\in I_{\mathrm{cal}}} \widetilde w_i \, d_{\mathrm{TV}}\!\bigl(\mathcal{L}(\boldsymbol{L}),\,\mathcal{L}(\boldsymbol{L}^{(i)})\bigr).
\]
where $\boldsymbol{L}$ is the vector of losses.
\end{theorem}

\begin{theorem}[ACI Guarantee, Prop. 4.1 in \cite{gibbs_adaptive_2021}]
\label{thm:aci_app}
ACI satisfies, with probability one,
\[
\lim_{|I_{\mathrm{test}}|\to\infty} \frac{1}{|I_{\mathrm{test}}|}\sum_{t \in I_{\mathrm{test}}} \mathbf{1}\{Y_t \notin {\mathcal C}_{1-\alpha_t}(X_t)\}
\ =\ \alpha,
\]
and for any finite test block length $|I_{\mathrm{test}}|$,
\[
\left|\;\frac{1}{|I_{\mathrm{test}}|}\sum_{t \in I_{\mathrm{test}}} \mathbf{1}\{Y_t \notin {\mathcal C}_{1-\alpha_t}(X_t)\} - \alpha \;\right|
\ \le\ \frac{\max\{\alpha_{m},\,1-\alpha_{m}\}+\gamma}{|I_{\mathrm{test}}|\gamma},
\]
where $m$ is the first test index.
\end{theorem}

\clearpage
\phantomsection

\bibliographystyle{plainnat}
\bibliography{references_fixed}

\end{document}